\documentclass[pra,aps,10pt,notitlepage,twocolumn,nofootinbib,superscriptaddress]{revtex4-2}

\usepackage{amsfonts,amsmath,amssymb,mathtools,mathrsfs,bbm,float}
\usepackage[ruled, vlined, linesnumbered]{algorithm2e}
\usepackage{graphicx,epic,eepic,epsfig,latexsym,verbatim,color}
\usepackage[inline,shortlabels]{enumitem}
\usepackage[caption=false]{subfig}
\usepackage[dvipsnames]{xcolor}
 
\usepackage{tikz}
\usetikzlibrary{chains}
\usetikzlibrary{fit}

\usepackage{epsfig}
\usetikzlibrary{shapes.symbols,patterns} % for source symbols
\usepackage{pgfplots}

\usepackage[strict]{changepage}

\usepackage[marginal]{footmisc}
\usepackage{url}
\usepackage{theorem}
\usepackage{thmtools}
\usepackage[ISO]{diffcoeff}
\usepackage{xfrac}

\newtheorem{proposition}{Proposition}
\newtheorem{lemma}[proposition]{Lemma}

\newtheorem{theorem}[proposition]{Theorem}

\newenvironment{proof}{\noindent \textbf{{Proof~} }}{\hfill $\blacksquare$}
\def\squareforqed{\hbox{\rlap{$\sqcap$}$\sqcup$}}
\def\qed{\ifmmode\squareforqed\else{\unskip\nobreak\hfil
\penalty50\hskip1em\null\nobreak\hfil\squareforqed
\parfillskip=0pt\finalhyphendemerits=0\endgraf}\fi}
\def\endenv{\ifmmode\;\else{\unskip\nobreak\hfil
\penalty50\hskip1em\null\nobreak\hfil\;
\parfillskip=0pt\finalhyphendemerits=0\endgraf}\fi}

\newcounter{example}

\mathchardef\ordinarycolon\mathcode`\:
\mathcode`\:=\string"8000
\def\vcentcolon{\mathrel{\mathop\ordinarycolon}}
\begingroup \catcode`\:=\active
  \lowercase{\endgroup
  \let :\vcentcolon
  }

\usepackage{hyperref}
\hypersetup{colorlinks=true,citecolor=blue,linkcolor=blue,filecolor=blue,urlcolor=blue,breaklinks=true}
\usepackage{footnotebackref}
\usepackage{cleveref}
\usepackage{graphicx}

\definecolor{darkblue}{RGB}{0,76,156}
\definecolor{darkkblue}{RGB}{0,0,153}
\definecolor{blue2}{RGB}{102,178,255}
\definecolor{darkred}{RGB}{195,0,0}

\RequirePackage[framemethod=default]{mdframed}
\newmdenv[skipabove=7pt,
skipbelow=7pt,
backgroundcolor=darkblue!15,
innerleftmargin=5pt,
innerrightmargin=5pt,
innertopmargin=5pt,
leftmargin=0cm,
rightmargin=0cm,
innerbottommargin=5pt,
linewidth=1pt]{tBox}

\newmdenv[skipabove=7pt,
skipbelow=7pt,
backgroundcolor=blue2!25,
innerleftmargin=5pt,
innerrightmargin=5pt,
innertopmargin=5pt,
leftmargin=0cm,
rightmargin=0cm,
innerbottommargin=5pt,
linewidth=1pt]{dBox}

\newmdenv[skipabove=7pt,
skipbelow=7pt,
backgroundcolor=darkred!15,
innerleftmargin=5pt,
innerrightmargin=5pt,
innertopmargin=5pt,
leftmargin=0cm,
rightmargin=0cm,
innerbottommargin=5pt,
linewidth=1pt]{rBox}

\newcommand{\nc}{\newcommand}
\nc{\bra}[1]{\langle#1|}
\nc{\ket}[1]{|#1\rangle}
\nc{\ketbra}[2]{\lvert#1\rangle\!\langle#2\rvert}
\nc{\braket}[2]{\langle#1|#2\rangle}

\DeclarePairedDelimiter{\norm}{\lVert}{\rVert}
\DeclarePairedDelimiter{\abs}{\lvert}{\rvert}

\DeclarePairedDelimiterX{\infdivx}[2]{(}{)}{%
  #1\;\delimsize\|\;#2%
}

\nc{\proj}[1]{| #1\rangle\!\langle #1 |}
\nc{\avg}[1]{\langle#1\rangle}
\nc{\smfrac}[2]{\mbox{$\frac{#1}{#2}$}}
\nc{\tr}{\operatorname{Tr}}
\nc{\ox}{\otimes}
\nc{\dg}{\dagger}
\nc{\dn}{\downarrow}
\nc{\cA}{{\cal A}}
\nc{\cB}{{\cal B}}
\nc{\cC}{{\cal C}}
\nc{\cD}{{\cal D}}
\nc{\cE}{{\cal E}}
\nc{\cF}{{\cal F}}
\nc{\cG}{{\cal G}}
\nc{\cH}{{\cal H}}
\nc{\cI}{{\cal I}}
\nc{\cJ}{{\cal J}}
\nc{\cK}{{\cal K}}
\nc{\cL}{{\cal L}}
\nc{\cM}{{\cal M}}
\nc{\cN}{{\cal N}}
\nc{\cO}{{\cal O}}
\nc{\cP}{{\cal P}}
\nc{\cQ}{{\cal Q}}
\nc{\cR}{{\cal R}}
\nc{\cS}{{\cal S}}
\nc{\cT}{{\cal T}}
\nc{\cU}{{\cal U}}
\nc{\cV}{{\cal V}}
\nc{\cX}{{\cal X}}
\nc{\cY}{{\cal Y}}
\nc{\cZ}{{\cal Z}}
\nc{\cW}{{\cal W}}
\nc{\csupp}{{\operatorname{csupp}}}
\nc{\qsupp}{{\operatorname{qsupp}}}
\nc{\var}{{\operatorname{var}}}
\nc{\rar}{\rightarrow}
\nc{\lrar}{\longrightarrow}
\nc{\polylog}{{\operatorname{polylog}}}
\nc{\wt}{{\operatorname{wt}}}
\nc{\supp}{{\operatorname{supp}}}

\nc{\argmin}{{\operatorname{argmin}}}

\makeatletter
\newcommand{\tpmod}[1]{{\@displayfalse\pmod{#1}}}
\makeatother

\def\i{\mathbf{i}}

\def\x{\xi}

\nc{\RR}{{{\mathbb R}}}
\nc{\CC}{{{\mathbb C}}}
\nc{\FF}{{{\mathbb F}}}
\nc{\NN}{{{\mathbb N}}}
\nc{\ZZ}{{{\mathbb Z}}}
\nc{\PP}{{{\mathbb P}}}
\nc{\QQ}{{{\mathbb Q}}}
\nc{\UU}{{{\mathbb U}}}
\nc{\EE}{{{\mathbb E}}}
\nc{\id}{{\operatorname{id}}}

\nc{\CHSH}{{\operatorname{CHSH}}}

\nc{\rU}{\mbox{U}}

\nc{\ob}[1]{#1}

\nc{\SEP}{{\text{\rm SEP}}}
\nc{\NS}{{\text{\rm NS}}}
\nc{\LOCC}{{\text{\rm LOCC}}}
\nc{\PPT}{{\text{\rm PPT}}}
\nc{\EXT}{{\text{\rm EXT}}}
\nc{\Sym}{{\operatorname{Sym}}}

\nc{\ERLO}{{E_{\text{r,LO}}}}
\nc{\ERLOCC}{{E_{\text{r,LOCC}}}}
\nc{\ERPPT}{{E_{\text{r,PPT}}}}
\nc{\ERLOCCinfty}{{E^{\infty}_{\text{r,LOCC}}}}
\nc{\Aram}{{\operatorname{\sf A}}}

\makeatletter
\def\grd@save@target#1{%
  \def\grd@target{#1}}
\def\grd@save@start#1{%
  \def\grd@start{#1}}
\tikzset{
  grid with coordinates/.style={
    to path={%
      \pgfextra{%
        \edef\grd@@target{(\tikztotarget)}%
        \tikz@scan@one@point\grd@save@target\grd@@target\relax
        \edef\grd@@start{(\tikztostart)}%
        \tikz@scan@one@point\grd@save@start\grd@@start\relax
        \draw[minor help lines,magenta] (\tikztostart) grid (\tikztotarget);
        \draw[major help lines] (\tikztostart) grid (\tikztotarget);
        \grd@start
        \pgfmathsetmacro{\grd@xa}{\the\pgf@x/1cm}
        \pgfmathsetmacro{\grd@ya}{\the\pgf@y/1cm}
        \grd@target
        \pgfmathsetmacro{\grd@xb}{\the\pgf@x/1cm}
        \pgfmathsetmacro{\grd@yb}{\the\pgf@y/1cm}
        \pgfmathsetmacro{\grd@xc}{\grd@xa + \pgfkeysvalueof{/tikz/grid with coordinates/major step}}
        \pgfmathsetmacro{\grd@yc}{\grd@ya + \pgfkeysvalueof{/tikz/grid with coordinates/major step}}
        \foreach \x in {\grd@xa,\grd@xc,...,\grd@xb}
        \node[anchor=north] at (\x,\grd@ya) {\pgfmathprintnumber{\x}};
        \foreach \y in {\grd@ya,\grd@yc,...,\grd@yb}
        \node[anchor=east] at (\grd@xa,\y) {\pgfmathprintnumber{\y}};
      }
    }
  },
  minor help lines/.style={
    help lines,
    step=\pgfkeysvalueof{/tikz/grid with coordinates/minor step}
  },
  major help lines/.style={
    help lines,
    line width=\pgfkeysvalueof{/tikz/grid with coordinates/major line width},
    step=\pgfkeysvalueof{/tikz/grid with coordinates/major step}
  },
  grid with coordinates/.cd,
  minor step/.initial=.2,
  major step/.initial=1,
  major line width/.initial=2pt,
}
\makeatother

\usepackage{thmtools}
\usepackage{thm-restate}
\usepackage{etoolbox}
\makeatletter
\def\problem@s{}
\newcounter{problems@cnt}

\newcommand{\allproblems}{\problem@s}
\makeatother

\usepackage{times}
\usepackage{amsmath,amsfonts,amssymb,graphicx,mathtools,bm,tcolorbox,relsize}
\usepackage[utf8]{inputenc}
\usepackage[T1]{fontenc}
\pgfplotsset{compat=1.18}
\usepackage{natbib}
\usepackage{multirow}
\usepackage{booktabs}
\usepackage{array}
\usepackage{cleveref}
\usepackage{subfig}

\usepackage{tabularx}
\usepackage{soul}
\usepackage[export]{adjustbox}

\definecolor{colortwo}{rgb}{0.4,0.77,0.17}
\definecolor{colorthree}{rgb}{0.01,0.51,0.93}
\definecolor{darkgray}{rgb}{0.3,0.3,0.3}

\usepackage{tikz}
\usetikzlibrary{quantikz2}
\usetikzlibrary{external} 
\usepackage{duckuments} % used for placeholders
\usepackage{soul}
\soulregister\cite7
\soulregister\ref7
\soulregister\pageref7

\nc{\add}[1]{{\color{BrickRed} #1}}
\nc{\remove}[1]{{\color{lightgray} \st{#1}}}
\nc{\authorcomment}[1]{{\begin{tcolorbox}[colback=green!5!white,colframe=green!10!gray,title=Comments from authors] #1 \end{tcolorbox}}}
\nc{\replace}[2]{ #2}

\newcommand{\trace}[2][]{\tr_{#1}\!\left[ #2 \right]}
\newcommand{\ddx}[1]{\frac{\textup{d}}{\textup{d}{\kern .08333em}#1}}
\newcommand{\integral}[4]{\int_{#1}^{#2} #4 {\kern .16666em}\textup{d}{#3}}

\newcommand{\set}[1]{ \left\{{\kern .02083em} #1 {\kern .02083em}\right\} }
\newcommand{\setcond}[2]{ \left\{{\kern .02083em} #1 :{\kern .08333em} #2 {\kern .02083em}\right\} }

\newcommand{\bmh}{{h}}
\newcommand{\bmj}{{\bm j}}
\newcommand{\bmm}{{\bm m}}

\newcommand{\fH}{{\mathfrak{H}}}

\newcommand{\SWAP}{{\mathtt{SWAP}}}
\newcommand{\CNOT}{{\mathtt{CNOT}}}
\newcommand{\CZ}{{\mathtt{CZ}}}

\newcommand{\gibbs}[2]{G_{#1, #2}}
\newcommand{\pauli}[2]{{\mathcal{E}_{#2}^{#1}}}
\newcommand{\pauliselect}[1]{\mathcal{N}_{#1}}
\newcommand{\measureselect}[1]{\mathcal{M}_{#1}}
\newcommand{\densityspace}[1]{\textup{D}\!\left(\mathcal{H}_{#1}\right)}

\newcommand{\dataset}{S(\beta, \Sigma, \bmh)}
\newcommand{\datasetN}{S_N(\beta, \Sigma, \bmh)}

\newcommand{\tnorm}[1]{\lVert #1 \rVert_1}
\newcommand{\inorm}[1]{\lVert #1 \rVert_\infty}
\newcommand{\comm}[2]{\left[#1,{\kern .16666em}#2\right]}

\newcommand{\prob}[1]{\textup{Pr}\!\left( #1 \right)}
\newcommand{\expect}[2][]{\mathbb{E}_{#1}\!\left[ #2 \right]}
\newcommand{\expectcond}[2]{\mathbb{E}\!\left[ #1 {\kern .16666em}|{\kern .16666em} #2 \right]}

\newcommand{\idop}{{\mathbbm{1}}}
\newcommand{\parity}[1]{p(#1)}

\newcommand{\poly}{\operatorname{poly}}
\newcommand{\BigO}[1]{{\mathcal{O}\!\left( #1 \right)}}
\newcommand{\BigTO}[1]{{\widetilde{\mathcal{O}}\!\left( #1 \right)}}
\newcommand{\BigOmega}[1]{{\Omega\!\left( #1 \right)}}

\newcommand{\hsum}{\lambda}

\newcommand{\loss}{\mathcal{L}}

\newcommand{\bfx}{\mathbf{x}}
\newcommand{\walk}[1]{\bfx^{(#1)}}
\newcommand{\walkH}[1]{H^{(#1)}}
\newcommand{\probQ}[1]{\textup{Pr}_{Q}\!\left( #1 \right)}

\newcommand{\specialcell}[2][c]{%
  \begin{tabular}[#1]{@{}l@{}}#2\end{tabular}}

\allowdisplaybreaks

\begin{document}
\title{Thermal-Drift Sampling: Generating Thermal Ensembles for Learning Many-Body Systems}

\author{Jiyu Jiang}
\author{Mingrui Jing}
\author{Jizhe Lai}
\author{Xin Wang}
\email{felixxinwang@hkust-gz.edu.cn}
\author{Lei Zhang}

\affiliation{Thrust of Artificial Intelligence, Information Hub,\\
The Hong Kong University of Science and Technology (Guangzhou), Guangzhou 511453, China}

\date{\today}

\begin{abstract}
Thermal equilibrium states of many-body Hamiltonians are essential for probing quantum chaos, finite-temperature phases of matter, and training quantum machine learning models, yet generating large collections of such states across different Hamiltonians remains costly with existing methods.
We introduce a powerful operation, the quantum thermal-drift channel, to construct a measurement-controlled sampling algorithm that autonomously generates thermal states together with their system Hamiltonians as labels for general physical models. 
We prove that our algorithm is efficient: the total gate count scales polynomially with system size and quadratically with inverse temperature, providing the first polynomial resource bound for random thermal state generation.
We characterize the distribution of sampled Hamiltonians as a normal distribution reweighted by partition functions, which quantifies a trade-off between sampling accuracy and effective label range.
Level-spacing statistics computed from sampled thermal states of a 2D transverse-field Ising model show a crossover to Wigner--Dyson universality, confirming that the sampler captures nontrivial chaotic correlations. 
Finally, a variational quantum classifier trained on the generated dataset achieves near-optimal accuracy in predicting Hamiltonian properties of unseen states. 
These results establish a scalable, quantum-native route for thermodynamic simulation and labeled quantum data generation in many-body systems.
\end{abstract}

\maketitle

%%%%%%%%%%%%%%%%%%%%%%%%%%%%%%%%%%%%%%%%%%%%%%%%%%%%%%%%%%%%%%%%%%%%%%%%%%%
\paragraph*{Introduction.---}
Random quantum states are a fundamental resource in quantum information science~\cite{khemani2017critical} which exhibit universal entanglement and spectral properties predicted by random matrix theory, making them indispensable for benchmarking quantum devices~\cite{dankert2009exacta,magesan2011scalable} and probing the limits of classical simulability~\cite{maziero2015random,odavic2023random}. 
They are essential for understanding quantum chaos~\cite{choi2023preparing} and black-hole physics~\cite{hayden2007blacka}, and serve as testbeds for quantum supremacy~\cite{boixo2018characterizing,arute2019quantum} and device characterization~\cite{liu2023observation}. 
Recently, the rise of quantum machine learning has created a new demand for these states~\cite{cerezo2022challenges}. Large collections of quantum states are treated as ``quantum data'' to learn the underlying patterns of many-body systems~\cite{jager2023universala,flam2022learning}.

\begin{figure}[t]
    \centering
    \includegraphics[width=\linewidth]{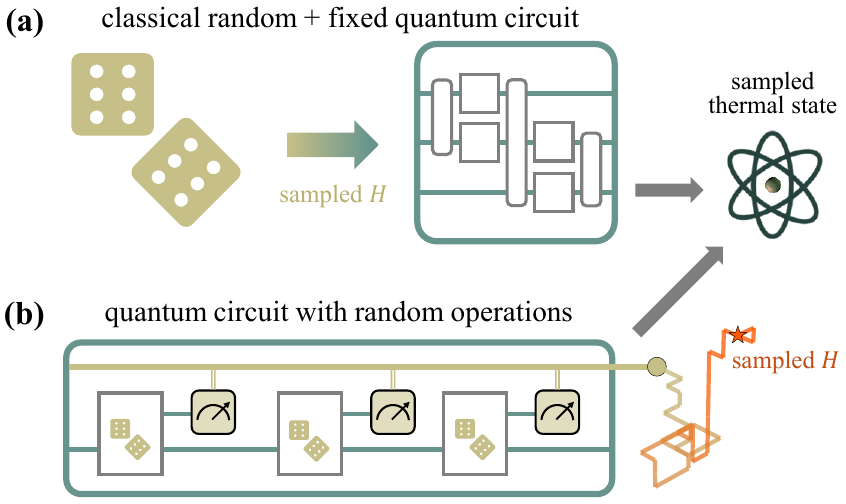}
    \caption{Difference between classical sampling and measurement-controlled sampling. The task is to sample a thermal state together with its Hamiltonian label.
    \textbf{(a)} Sampling the label via classical computers. The thermal state is prepared using existing thermal-state preparation circuits.
    \textbf{(b)} Sampling the thermal state and label via one structured random circuit. The label is generated by a lattice walk based on outcomes of mid-circuit measurements.
    }
    \label{fig:c vs q}
\end{figure}

The random-state families studied so far, such as Haar-random states~\cite{ji2018pseudorandom}, unitary $t$-designs~\cite{emerson2003pseudo}, and random stabilizer states~\cite{garcia2017geometry}, consist of pure states whose statistical properties are determined by the geometry of Hilbert space alone, without reference to a specific physical Hamiltonian.
Thermal states describe realistic many-body systems at finite temperature, exhibiting highly nontrivial correlations and often volume-law entanglement, making them representative of the ``typical'' mixed states encountered in quantum chaos and finite-temperature phases~\cite{liu2023observation}. They allow us to probe the eigenstate thermalization hypothesis~\cite{d2016quantum} and spectral statistics governed by random matrix theory~\cite{khemani2017critical}. 
Random thermal states therefore constitute a physically grounded family for probing universal many-body phenomena, motivating the development of efficient generation methods that do not rely on detailed spectral knowledge nor classical simulation.
Beyond preparing a single thermal state, data-driven applications such as Hamiltonian learning and quantum device benchmarking demand a large, labeled collection: many thermal states, each paired with the Hamiltonian from which it was generated. Producing such paired data $(\gibbs{\beta}{H}, H)$ at scale is a important challenge for supervised learning with quantum data.

With quantum processors advancing toward the early fault-tolerant regime, efficient algorithms for thermal state preparation on quantum hardware are increasingly sought after.
Practical generation of random thermal states, however, remains costly with existing both classical and quantum algorithms. 
Classical methods either suffer from the sign problem~\cite{troyer2005computational}  or exhibit volume-law entanglement that makes classical contraction costs grow exponentially with system size~\cite{verstraete2006matrix}.
Current quantum methods rely on quantum phase estimation combined with imaginary-time evolution~\cite{poulin2009sampling, chowdhury2017quantum, andras2019quantumsingular, motta2020determining, rouze2024optimal}, fluctuation theorems~\cite{Holmes2022quantumalgorithms}, or quantum Metropolis and Lindbladian thermalization algorithms~\cite{temme2011quantum, yung2012quantumquantum, moussa2019low, Rall2023thermalstate, shtanko2023preparing, wocjan2023szegedy, chen2023quantumthermal, chen2025efficient, chen2025efficientquantum, ding2025efficient, zhang2023dissipative, brandao2019finite} that couple the system to an effective bath. 
These instance-oriented methods can be effective for preparing one thermal state, but do not directly address the task of sampling thermal states from a family of Hamiltonians. Extending them to random thermal ensembles requires reconfiguring the circuit for each Hamiltonian instance, incurring substantial classical overhead that limits scalability.

In this Letter, we solve this critical problem by introducing a quantum algorithm for the scalable generation of random thermal states with provable polynomial efficiency, as illustrated in Figure~\ref{fig:c vs q}. 
The device is programmed solely with a target set of $L$ Pauli operators $\Sigma = \set{\sigma_1, \ldots, \sigma_L}$ representing the interaction topology. 
Our algorithm autonomously generates the thermal state $\gibbs{\beta}{H} = e^{-\beta H} / \trace{e^{-\beta H}}$ together with its label $H = \sum c_j \sigma_j$, with coefficients $c_j$ drawn from a thermal-weighted normal distribution, effectively functioning as a stochastic subroutine. 
Conceptually, the circuit performs a measurement-driven random walk in Hamiltonian space while approximately thermalizing the system state; the same measurement record that steers the evolution also specifies the Hamiltonian label.

Our algorithm operates as a measurement-controlled sampling protocol that applies a sequence of thermal-drift channels to generate the target thermal state. Unlike previous instance-oriented approaches~\cite{andras2019quantumsingular, chen2025efficient}, it provides the first explicit polynomial scaling of resource complexity with respect to all problem parameters in an elementary-gate model, as summarized in Table~\ref{tab:comparison}. 
The thermal-drift channel introduced here converts mid-circuit measurement randomness directly into physically meaningful Hamiltonian parameters, unifying state preparation and Hamiltonian sampling in a single quantum protocol.
Systematic benchmarking validates the predicted $\beta$-dependence and sampling efficiency. We further present two complementary applications of the sampler. First, level-statistics analysis confirms that the sampled states capture the spectral correlations characteristic of a chaotic regime. Second, as a proof of concept for quantum data generation, a variational quantum classifier trained on the sampled dataset achieves near-optimal accuracy in classifying unseen thermal states, illustrating that the sampler can directly provide labeled quantum datasets for learning tasks.

\begin{table}[t]
\centering
\caption{Resource complexity for generating a thermal state together with its Hamiltonian label $H$.
Here $n$, $\beta$, $\hsum$, $\epsilon$, and $\delta$ denote the number of qubits, inverse temperature, sum of coefficient bounds, target precision, and failure probability. $t_\textrm{mix}$ denotes the mixing time of a $\beta$-dependent Markov chain determined by the sampled $H$. Note that the prior methods were designed for preparing the thermal state of a single given Hamiltonian; the listed complexities reflect the cost of combining them with a classical label-sampling step, rather than producing $(\gibbs{\beta}{H},H)$ as an ensemble-level generation algorithm.}\label{tab:comparison}
\setlength{\tabcolsep}{1em}
\resizebox{\linewidth}{!}{
\begin{tabular}{llc}
    \toprule
    Workflow  & Query model & Complexity \\
    \midrule
    \multirow{3}{*}{\specialcell{Sample label + \\ prepare state}} & \specialcell{Block encoding of \\ normalized $H$~\cite{andras2019quantumsingular}} & $\BigO{2^{n/2}\hsum \beta / \trace{e^{-\beta H}}}$  \\
    \cmidrule{2-3}
    {}  & \specialcell{Block encodings of \\ jump operators~\cite{chen2025efficient}} & $\BigO{\beta t_\textrm{mix}}$  \\
    \midrule
    \specialcell{Sample state \\ with label} & \specialcell{Elementary gates \\ (this work)} & $\BigO{n^2 \lambda^2 \beta^2\epsilon^{-2/3}\log(1/\delta)}$ \\
    \bottomrule
\end{tabular}
}
\end{table}

\paragraph*{Generation of random thermal states.---} 
The task is to sample a thermal ensemble which is composed of a thermal state $\gibbs{\beta}{H}$ as ``data'' together with its Hamiltonian $H$ as ``label'', and $H$ is from a structured family of Pauli Hamiltonians.
Given a set of $L$ distinct $n$-qubit Pauli operators $\Sigma$ and the inverse temperature $\beta>0$, such family is described by a set of all linear combinations of operators in $\Sigma$ with bounded coefficients as
\begin{equation}
    \fH = \setcond{\sum\nolimits_{j} c_j \sigma_j}{\abs{c_j} \leq \bmh_j, c_j \in \RR, \sigma_j \in \Sigma}
.\end{equation}
where $h_j$'s are the boundary values of the coefficients. 
Unlike the pure-state families discussed in the introduction,
our setting asks for generating a structured thermal state
that is linked to an explicit physical Hamiltonian. 
This provides a controllable generating set essential for physical benchmarking and Hamiltonian learning.

This setting imposes no restrictions on the locality or commutativity of the interactions.
Our verification experiments consider two physical models on a $3 \times 3$ two-dimensional grid, where two-local Pauli operators act on nearest-neighbor pairs. The first is a Heisenberg model with $\Sigma$ containing operators of the form $XX$, $YY$, and $ZZ$; the second is a transverse-field Ising model with $\Sigma$ containing operators of the form $ZZ$ and $X$ (or $XX$ and $Z$).
For both models, the coefficient bounds $h_j$ are uniform, and we denote $\hsum = \sum_j h_j$, which upper bounds the largest eigenvalue of any $H \in \fH$. 

\paragraph*{Quantum thermal-drift channel.---} 
Ground state preparation is known to be QMA-hard in the worst case~\cite{watrous2008quantum,chowdhury2020variationala}, and low-temperature thermal state preparation faces comparable complexity barriers~\cite{aharonov2013quantum}.
A standard `sample-then-prepare' pipeline, which first selects a Hamiltonian $H$ and subsequently attempts to prepare its thermal state, is therefore often prohibitively costly. To circumvent this, we implement quantum operations involving randomness directly on the quantum hardware.

We define the thermal-drift channel $\pauliselect{\sigma}$ for any Pauli operator $\sigma$, which is the core component of our algorithm. For any input state $\rho$, the channel is defined as
\begin{equation}~\label{eqn:thermal drift}
\pauliselect{\sigma}(\rho)
= \frac{\mu}{2}\big(  \ketbra{\uparrow}{\uparrow} \ox \pauli{\uparrow}{\sigma}(\rho) + 
 \ketbra{\downarrow}{\downarrow} \ox \pauli{\downarrow}{\sigma}(\rho) \big)
,\end{equation}
where $\mu < 1$ is a scaling factor. Here, for any state $\rho$, the maps $\pauli{\uparrow}{\sigma}$ and $\pauli{\downarrow}{\sigma}$ are two completely positive operations defined as $e^{+\tau\sigma/2} \rho\, e^{+\tau\sigma/2}$ and $e^{-\tau\sigma/2} \rho\, e^{-\tau\sigma/2}$, respectively,
with a thermal strength $\tau > 0$.
The first register of $\pauliselect{\sigma}(\rho)$ induces a random branch.
Measuring this flag register indicates which direction the state is thermalized toward, with probabilities proportional to the corresponding unnormalized traces.

The thermal-drift channel can be implemented using system-ancilla coupling in a dilated Hilbert space, followed by measuring and tracing out the flag register, such that the whole channel can be implemented within $\BigO{n}$ quantum elementary gates and $2n$ ancilla qubits.
We refer to the Supplemental Material for a detailed implementation.

\begin{figure*}[t]
    \centering
    \includegraphics[width=1\linewidth]{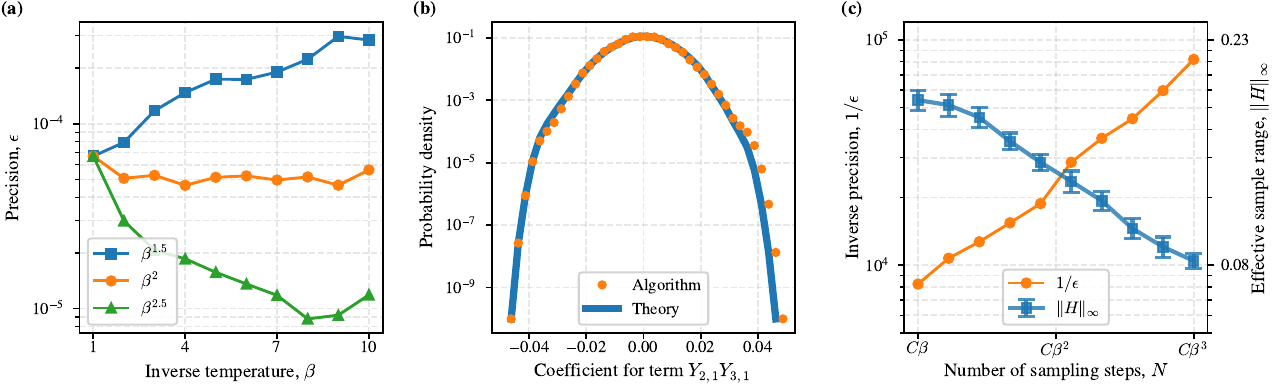}
    \caption{
Numerical verification of the theoretical predictions for the thermal-drift sampling algorithm applied to a $3\times 3$ two-dimensional Heisenberg model. 
See the Supplemental Material for additional experimental details.
\textbf{(a)} Precision $\epsilon$ between the output state and the ideal thermal state as a function of the inverse temperature $\beta$, for different step scalings. 
\textbf{(b)} Empirical marginal distribution of a nearest-neighbor $YY$ interaction coefficient, compared with the theoretical prediction from Theorem~\ref{thm:distribution}. Here $Y_{i,j}$ denotes the Pauli-$Y$ on the site at row $i$ and column $j$. 
\textbf{(c)} Trade-off between inverse precision and effective sample range at fixed $\beta$, as the number of steps $N$ increases from $C\beta$ to $C\beta^{3}$, where $C$ is a fixed constant.
    }
    \label{fig:num exp}
\end{figure*}

\paragraph*{Thermal state generation algorithm.---} Our algorithm is implemented as a circuit of $N$ blocks applied to an initial state $\rho_0$.
In block $k$, a thermal-drift channel $\pauliselect{\sigma_{j_k}}$ is applied, where $\sigma_{j_k}$ is selected independently from the set $\Sigma$. The probability of choosing $\pauliselect{\sigma_j}$ is weighted by the bound $\bmh_j$, so that the distribution is $\bmh_j/\hsum$. The thermal strength is fixed across all blocks to $\tau = \beta\hsum /N$.
After this thermal-drift channel, a measure-and-trace operation $\measureselect{m_k}$ is applied to the flag register, generating a measurement outcome $m_k$. This outcome is recorded as $m_k \in \set{\mathord{+}1,\mathord{-}1}$ corresponding to $\uparrow$ and $\downarrow$.

The full circuit implementation is labeled by an ordered list of values $\bmj = \set{j_1, \ldots, j_N}$ and an ordered list of directions $\bmm = \set{m_1, \ldots, m_N}$, that corresponds to a channel $\pauli{}{\bmj,\bmm}$.
In particular, one can collect the accumulated drift along each Pauli axis $\sigma_j$, which defines the overall Hamiltonian label
\begin{equation}
    H_{\bmj,\bmm} = \frac{\hsum}{N} \sum_{j=1}^N \sum_{k:j_k = j} m_k \sigma_j
.\end{equation}
When $\rho_0$ is the maximally mixed state, the output state $\pauli{}{\bmj,\bmm}(\rho_0)$ is an approximate copy of its thermal state $\gibbs{\beta}{H_{\bmj,\bmm}}$, thereby completing the task of random thermal state generation.

\begin{figure*}[t]
    \centering
    \includegraphics[width=\linewidth]{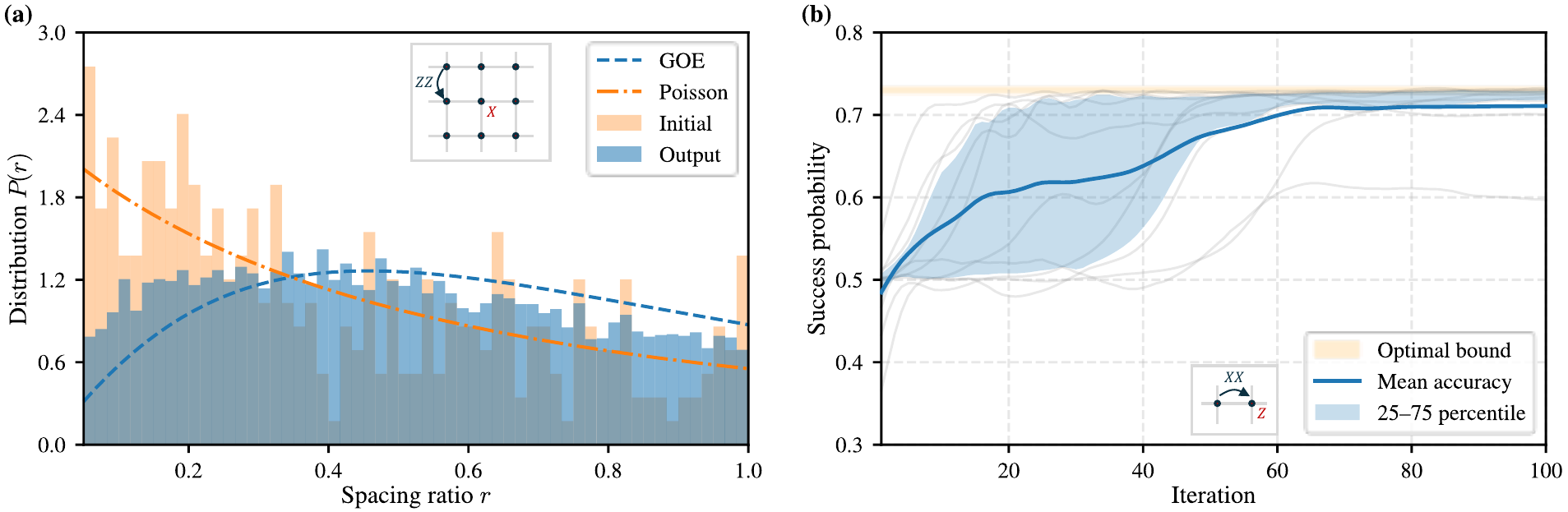}
    \caption{Applications of our thermal state sampler to level statistics and state classification. See the Supplemental Material for additional experimental details.
    \textbf{(a)} Level-spacing ratio statistics of the sampled thermal states for a $3 \times 3$ two-dimensional transverse-field Ising model.
    The orange histogram in the background corresponds to the initial state, while the blue histogram in the foreground corresponds to the output state generated by the sampling algorithm. The dashed curve indicates the GOE prediction, and the dash-dotted curve indicates the Poisson prediction. See the \emph{Application to signatures of quantum chaos} section for task details.
    \textbf{(b)} Classification accuracy of a quantum classifier trained on sampled thermal states of a two-qubit Ising model.
    Gray traces show individual trial runs, the blue curve shows the mean success probability across independent test sets over $10$ trials, the shaded region marks the $25$--$75$ percentile band, and the orange band indicates the range of theoretical upper bound on success probability. See the \emph{Application to state classification} section for task details.
    } 
    \label{fig:applications}
\end{figure*}

\paragraph*{Theoretical guarantees of sampling process.---} 
Unlike randomized algorithms such as QDRIFT~\cite{campbell2019random}, where the randomness is Markovian and can be analyzed directly in channel form, each measurement outcome in our algorithm depends on the evolving state, and hence the process is generally non-Markovian. 
We apply matrix-martingale concentration inequalities to bound the sampling error, and lattice random-walk theory to characterize the induced distribution over Hamiltonian labels.

We start with the sampling error $\epsilon$, quantified by the trace distance between the sampled state and the ideal thermal state. 
Given the sampled outcomes $\bmj$ and $\bmm$, we consider the effective Hamiltonian $H_\approx$ such that $\pauli{}{\bmj,\bmm}(\rho_0) = e^{-\beta H_\approx}/\trace{e^{-\beta H_\approx}}$.
Then the sampling error is bounded by $\epsilon \leq \beta\inorm{ H_{\bmj,\bmm} - H_\approx}$.
When the number of steps $N$ is not large enough, the sampling error behaves like the discretization error of a second-order Trotter product formula~\cite{childs2018toward}, scaling as $\BigO{\beta^3}$. However, when $N$ is sufficiently large so that Pauli operators are likely repetitive, we treat $\bmm$ as random variables conditional on a fixed $\bmj$. In this regime, the Freedman inequality for matrix martingales~\cite{tropp2011freedman} implies an $\BigO{\beta^2}$ scaling almost surely. 
This is numerically verified in Figure~\ref{fig:num exp}(a). For increasing inverse temperature starting at $\beta = 1$, we test our algorithm with number of steps $N = \BigO{\beta^{k}}$ for $k = 1.5$, $2$ and $2.5$, and observe that the average sampling error scales as $\BigO{\beta^{2-k}}$, which is consistent with the following result. Here $\BigTO{\cdot}$ omits the logarithm terms. 

\begin{theorem}[Sampling precision]~\label{thm:error}
    For any system size $n$, inverse temperature $\beta$, and sufficiently large $N$, our algorithm uses $\BigO{nN \log(1/\delta)}$ gates to sample a thermal state and its label, with precision $\BigTO{n^{3/2} \hsum^3 \beta^3 N^{-3/2}}$ and failure probability at most $\delta$.
\end{theorem}

This result guarantees that the gate resources required by the sampler scale polynomially in all problem parameters. 
The cubic dependence on $\beta$ reflects the accumulation of discretization error in the second-order Trotter product formula.
Then as the number of thermal-drift steps $N$ increases, the trace-distance error between the sampled state and the target thermal state decreases as $N^{-3/2}$, providing a route to higher-fidelity sampling.

To obtain a theoretical analysis of the sample distribution, 
we consider an asymptotic regime where the number of sampling steps $N$ tends to infinity, so that every possible label $H \in \fH$ can be sampled.
Then our algorithm can be modeled as a random walk $(\walk{k})_k$ on an $L$-dimensional integer lattice $\ZZ^L$, such that $\walk{0}$ and $\walk{N}$ correspond to the identity matrix and the sampled label $H$, respectively.
Such a walk sequence satisfies for all $k \geq 1$,
\begin{equation}
    \prob{\walk{k} - \walk{k-1} = \pm e_j} = \prob{m_k = \pm 1} \cdot \bmh_j/\hsum
,\end{equation}
where $e_j$ denotes the $j$-th unit vector. Applying lattice-based random walk theory~\cite{spitzer2001principles} yields a Gaussian approximation to the endpoint distribution under appropriate scaling. Combined with the thermal reweighting from $\prob{m_k = \pm 1}$, the sampling rule takes a product form involving the partition function and a normal density.

\begin{theorem}[Label distribution]~\label{thm:distribution}
    For any inverse temperature $\beta$, a label $H \in \fH$ can be asymptotically sampled by our algorithm with probability proportional to $\trace{e^{-\beta H}} D(\bfx)$, where $D$ is an $L$-dimensional multivariate normal distribution and $\bfx$ satisfies $\bfx_j = \trace{H \sigma_j}/\hsum \cdot N$.
\end{theorem}

The sampling probability is governed by two factors. The partition-function weight $\trace{e^{-\beta H}}$ biases the sampler toward Hamiltonians with more thermally accessible states, i.e., those with lower free energy at the given temperature. The normal density $D(\bfx)$ arises from the central-limit behavior of the underlying lattice random walk and concentrates the sampled coefficients near the origin. Together, they define a physically motivated measure over the Hamiltonian family.

We numerically show that a similar result holds in non-asymptotic cases. Taking $\beta = 2$, we sample from our algorithm and visualize one-dimensional slices of the marginal distribution of a chosen Pauli coefficient. As shown in Figure~\ref{fig:num exp}(b), the empirical distribution closely matches the ideal one. Note that the difference mainly occurs in both ends of the distribution, which will get smaller as the number of sampling step further increases.

We apply Theorems~\ref{thm:error} and~\ref{thm:distribution} to identify a trade-off between the sampling accuracy and the effective range of sampled Hamiltonians.
Due to the nature of random walks, higher accuracy comes with a more concentrated label distribution, and vice versa. As shown in Figure~\ref{fig:num exp}(c), for $\beta = 2$, as the number of steps $N$ increases from $\BigO{\beta}$ to $\BigO{\beta^3}$, the inverse precision $1/\epsilon$ increases while the effective sample range $\inorm{H}$ decreases at a similar rate. Therefore, despite the efficiency of a single execution, it is not recommended to use our algorithm for preparing thermal states of one fixed Hamiltonian. However, this characteristic makes our algorithm particularly well-suited for ensemble-level applications.

\paragraph*{Application to signatures of quantum chaos.---} 
Level statistics characterizes thermodynamic phases in quantum many-body systems~\cite{potter2015universal} by probing spectral correlations of quantum states. Applying this analysis to thermal states across different Hamiltonians requires generating those states efficiently, a capability that our sampler can provide. 
Given any random density operator with an ordered many-body eigenvalue spectrum $\set{\lambda_j}_j$, one defines the adjacent logarithmic level spacings $\delta_j = \log(\lambda_j/\lambda_{j+1})$ and the dimensionless spacing (gap) ratio~\cite{oganesyan2007localization,sierant2020model},
\begin{equation}
    r_j = \frac{\min\{\delta_j, \delta_{j+1}\}}{\max\{\delta_j, \delta_{j+1}\}} \in [0,1],
\end{equation}
whose distribution is insensitive to spectral unfolding and thus widely used in finite-size studies~\cite{serbyn2016spectral,buijsman2019random,odavic2023random}. In ergodic systems satisfying the eigenstate thermalization hypothesis, strong level repulsion leads to Wigner--Dyson statistics~\cite{wigner1967random,buijsman2019random}, whereas in the many-body localization phase, emergent local integrals of motion suppress correlations and produce Poisson level statistics~\cite{pal2010manybody,khemani2017critical}. The Wigner--Dyson distribution describes the probability of the spacing ratio in the Gaussian orthogonal ensemble (GOE), serving as the universal signature of chaos for quantum systems that preserve time-reversal symmetry~\cite{rabson2004crossover,serbyn2016spectral}.

To assess the chaotic nature of the generated states, we apply our algorithm to the level statistics of the sampled thermal state ensemble. An initial reference state whose spectrum exhibits near-Poisson statistics is compared with the output states produced by the sampler.
As shown in Figure~\ref{fig:applications}(a), aggregating $\set{r_j}_j$ over samples yields a robust, unfolding-free diagnostic. The level statistics shift from a Poissonian distribution in the reference state to a Wigner--Dyson distribution in the output thermal ensemble, confirming that the sampler captures the spectral correlations of the ergodic regime. 
The thermal-drift mechanism thus drives the system into a substantially more chaotic sector of Hilbert space than the structured input, producing states whose spectral correlations are consistent with physical thermalization.

\paragraph*{Application to state classification.---} 
Our thermal state sampler provides a scalable data source for quantum supervised learning algorithms~\cite{artymowicz2024efficient, chen2025learning, gu2024practical} that learn properties of quantum systems from their thermal states.
Consider classifying a black-box quantum system with an unknown Hamiltonian $H \in \fH$ by a label function $f(H) \in \set{0, 1}$, given access to only a single copy of its thermal state.
A natural approach is to deploy a variational quantum classifier~\cite{jager2023universala,flam2022learning,benedetti2019parameterized} (VQC) on a large labeled dataset of thermal states. Theorem~\ref{thm:error} implies that our sampler can generate such a dataset with cost polynomial in system size and inverse temperature, bypassing the need for individual state preparation. In the expectation-value-based formulation, training typically requires repeatedly measuring the same input state to estimate observables; here, we instead use single-shot measurements on freshly regenerated batches from the sampler, which is naturally compatible with a stochastic data source.

As a proof of concept, a VQC is trained to classify two-qubit Ising models described by $\Sigma = \set{XX, ZI, IZ}$ and $\bmh = (1, 1, 1)$ according to the sign of the $XX$ coefficient, where the optimal success probability is upper bounded by the Helstrom limit~\cite{helstrom1969quantum}. The classifier is optimized via a stochastic variant of the parameter-shift rule~\cite{mitarai2018quantum,schuld2019evaluating}. In our implementation, each gradient evaluation uses single-shot measurements on a freshly regenerated batch of $10^3$ approximate thermal states from our sampler. The classifier is evaluated, in each trial, on an independently generated test set of $3 \times 10^3$ exact thermal states. 
As shown in Figure~\ref{fig:applications}(b), the prediction accuracy approaches the Helstrom limit in most trials, with noticeable trial-to-trial variability.
In particular, the classifier trains on approximate states sampled from the thermal-weighted normal distribution (Theorem~\ref{thm:distribution}) yet generalizes to test states drawn from a uniform distribution. This distribution-shift robustness is enabled by the polynomial sampling cost (Theorem~\ref{thm:error}), which allows generating sufficiently large and diverse training sets without individual state preparation. Combined with the built-in Hamiltonian labeling, the sampler serves as an end-to-end data source for quantum supervised learning in this setting. See the Supplemental Material for training details.

\paragraph*{Conclusion and outlook.---}
We have introduced the quantum thermal-drift channel and used it to develop a measurement-controlled algorithm that generates random thermal states for Hamiltonians drawn from general Pauli models, with sampling cost polynomial in both the system size and the inverse temperature. This efficiency bypasses the high overhead of instance-oriented preparation methods and provides, to our knowledge, the first polynomial-cost quantum protocol for random thermal state generation. Two applications demonstrate the versatility of this approach: level-statistics analysis confirms that the sampled states exhibit Wigner--Dyson universality, a regime associated with strong level repulsion and quantum chaos; and a proof-of-concept classification experiment shows that the sampled dataset can serve as a data source for quantum learning tasks.

Several avenues for future research remain open. Our protocol relies on elementary gates and mid-circuit measurements, and can be viewed as a stochastic subroutine for generating labeled quantum data on early fault-tolerant devices. 
Improving the system-size scaling, potentially by exploiting structure in physically motivated Hamiltonian families, remains an important direction. We note that the polynomial efficiency achieved here contrasts with the known limitations of leading classical approaches for generic models: quantum Monte Carlo encounters the sign problem for frustrated systems~\cite{troyer2005computational}, and tensor-network contractions incur costs that grow exponentially with the volume-law entanglement present in generic thermal states~\cite{verstraete2006matrix}. Formalizing this contrast into a rigorous quantum advantage proof remains an interesting open question. 
On the machine learning side, extending the classification framework to more complex Hamiltonians and scalable optimization strategies would further establish the practical utility of quantum-generated thermal state datasets.

\paragraph*{Note added.---} We recently became aware of related independent works~\cite{gomez2026pauli,rudolph2026thermal} that simulate thermal states by applying sequences of Pauli-based nonunitary updates starting from the maximally mixed state. These works focus on thermal state preparation for a given Hamiltonian via classical algorithms, whereas we focus on thermal state generation for random Hamiltonians via quantum circuits.

\paragraph*{Acknowledgment.---}
The authors are listed in alphabetical order. 
%The authors would like to thank xxx for thoughtful and insightful comments. 
This work was partially supported by the National Key R\&D Program of China (Grant No.~2024YFB4504004), the National Natural Science Foundation of China (Grant No.~12447107, 92576114), the Guangdong Provincial Quantum Science Strategic Initiative (Grant No.~GDZX2403008, GDZX2503001), the CCF-Tencent Rhino-Bird Open Research Fund, and the Guangdong Provincial Key Lab of Integrated Communication, Sensing and Computation for Ubiquitous Internet of Things (Grant No.~2023B1212010007).

\paragraph*{Code availability.---} Numerical experiments are based on an open-source Python research software for quantum computing~\cite{quairkit}. Code and data used in the numerical experiments are available on \url{https://github.com/QuAIR/SamplingThermalState-Codes}.

%%%%%%%%% SUPPLEMENTAL MATERIAL %%%%%%%%%

\clearpage
\appendix
\setcounter{subsection}{0}
\setcounter{table}{0}
\setcounter{figure}{0}

\vspace{3cm}
\onecolumngrid
% \vspace{2cm}

\begin{center}
\Large{\textbf{Supplemental Material for } \\ \textbf{Thermal-Drift Sampling: Generating Thermal Ensembles for Learning Many-Body Systems}}
\end{center}

% Appendix numbering
% \renewcommand{\theequation}{S\arabic{equation}}
\numberwithin{equation}{section}
\renewcommand{\theproposition}{S\arabic{proposition}}
\renewcommand{\thelemma}{S\arabic{proposition}}
\renewcommand{\thetheorem}{S\arabic{proposition}}
\renewcommand{\thedefinition}{S\arabic{definition}}
\renewcommand{\thefigure}{S\arabic{figure}}

\providecommand{\theHlemma}{}
\renewcommand{\theHlemma}{S\arabic{proposition}}

\renewcommand{\theHequation}{S\arabic{equation}}
\renewcommand{\theHtheorem}{S\arabic{proposition}}
\renewcommand{\theHtable}{S\arabic{table}}
\renewcommand{\theHfigure}{S\arabic{figure}}

% Reset counters
\setcounter{equation}{0}
\setcounter{table}{0}
\setcounter{section}{0}
\setcounter{proposition}{0}
\setcounter{definition}{0}
\setcounter{figure}{0}

% \tableofcontents
% %%%%%%%%%%%%%%%%%%%%%%%%%%%%%%%%%%%%%%%%%%%%%%%%%%%%%%%%%%%%%%%%%%%%%%%%%%%

%%%%%%%%%%%%%%%%%%%%%%%%%%%%%%%%%%%%%%%%%%%%%%%%%%%%%%%%%%%%%%%%%%%%%%%%%%%%

In this supplementary material, we offer detailed proofs of the theorems and propositions in the manuscript `Thermal-Drift Sampling: Generating Thermal Ensembles for Learning Many-Body Systems.' In section~\ref{sec:preliminaries}, we will introduce the notations and the basic lemmas that we used in the proofs and the analysis. We then, introduce the details of our thermal-drift channel in section~\ref{sec:tdc} which serves as the core of our algorithm. After that, the fundamental setups of the random thermal state generation task and the algorithm statements are clarified in section~\ref{sec:sample proof}. We also provide performance and error analysis of our algorithm in the subsections of the same section. The applications towards the level statistics and the state classification are discussed in section~\ref{sec:apps}. The corresponding experimental details are provided in the followed section~\ref{sec:exp_setting} for board interests.

\section{Notations and Preliminaries}\label{sec:preliminaries}
\paragraph*{Asymptotic notations.} Let $[N] = \{1,2,\dotsc,N\}$ for any positive integer $N$.
We use asymptotic notations to describe the scaling of complexity and error bounds. 
For two non-negative functions $f(x)$ and $g(x)$, we write $f(x) = \BigO{g(x)}$ if there exist constants $C > 0$ and $x_0$ such that $f(x) \leq C g(x)$ for all $x \geq x_0$. 
Conversely, $f(x) = \BigOmega{g(x)}$ indicates that $f(x) \geq c g(x)$ for some constant $c > 0$ and sufficiently large $x$. 
The notation $f(x) = o(g(x))$ implies that $\lim_{x \to \infty} f(x)/g(x) = 0$.
We also use the soft-O notation $\BigTO{g(x)}$ to suppress polylogarithmic factors, i.e., $\BigTO{g(x)} \equiv \BigO{g(x) \poly(\log g(x))}$.

\paragraph*{Operators.} We denote the single-qubit Pauli operators by the set $\set{I, X, Y, Z}$, defined as $I = \begin{psmallmatrix} 1 & 0 \\ 0 & 1 \end{psmallmatrix}$, $X = \begin{psmallmatrix} 0 & 1 \\ 1 & 0 \end{psmallmatrix}$, $Y = \begin{psmallmatrix} 0 & -i \\ i & 0 \end{psmallmatrix}$ and $Z = \begin{psmallmatrix} 1 & 0 \\ 0 & -1 \end{psmallmatrix}$.
The set of $n$-qubit Pauli operators is denoted as $\set{I, X, Y, Z}^{\ox n}$.
We denote the Hadamard gate and the single-qubit rotation around the $y$-axis by: $H = \frac{1}{\sqrt{2}}\begin{psmallmatrix} 1 & 1 \\ 1 & -1 \end{psmallmatrix}$ and $R_y(\theta) = e^{-i \theta Y / 2} = \begin{psmallmatrix} \cos(\theta/2) & -\sin(\theta/2) \\ \sin(\theta/2) & \cos(\theta/2) \end{psmallmatrix}$.
The commutator of operators $A$ and $B$ is $\comm{A}{B}=AB-BA$.
The adjoint action (Lie bracket) is defined recursively as $\operatorname{ad}_A(B) = \comm{A}{B}$ and $\operatorname{ad}^k_A(B) = \comm{A}{\operatorname{ad}^{k-1}_A(B)}$ for $k \geq 2$.

\paragraph*{Quantum systems.} Let $\cH$ denote a finite-dimensional complex Hilbert space. 
A density operator $\rho$ is a linear operator acting on $\cH$ that is positive semi-definite matrix with unit trace, and we write $\densityspace{A}$ for the set of density operators on a system $A$ with Hilbert space $\cH_A$.
For a composite system with Hilbert space $\cH_{AB} = \cH_A \ox \cH_B$, we write $X_A = X \ox I_B$ to denote an operator $X$ acting locally on subsystem $A$, where $I_B$ is the identity on $\cH_B$.

\paragraph*{Norms.} For a vector $\bfx=(\bfx_1, \dotsc, \bfx_L) \in \CC^L$, we define the Euclidean norm as $\norm{\bfx} = \sqrt{\sum_{i=1}^L \abs{\bfx_i}^2}$.
For a matrix $A \in \CC^{d \times d}$, we define the spectral norm $\inorm{A}$ as the largest singular value of $A$, and the trace norm $\tnorm{A}$ as the sum of the singular values of $A$.

\paragraph*{Parity functions.} For $N\in\NN$ and $\bfx=(\bfx_1,\ldots,\bfx_L)\in\ZZ^L$, define the \emph{reachability parity function} $a_N: \ZZ^L\to \set{0, 2}$ by
\begin{equation}~\label{eqn:reach parity def}
    a_N(\bfx) =
    \begin{cases}
    2, & \textrm{if }\sum_{j=1}^L |\bfx_j|\leq N
           \text{ and } N - \sum_{j=1}^L |\bfx_j| \textrm{ is even};\\
    0, & \textrm{otherwise.}
  \end{cases}
\end{equation}
Thus $a_N (\bfx)=0$ whenever $\bfx$ is not reachable from the origin in $N$ steps of the walk, and $a_N(\bfx)=2$ otherwise.
Define the \emph{bitwise parity function} for $x \in \NN$ as
\begin{equation}~\label{eqn:parity def}
\parity{x} = \begin{cases}
    0 , &\textrm{if $\sum_i x_i$ is even};\\
    1, &\textrm{otherwise,}
\end{cases}
\end{equation}
where $x_{N - 1} \cdots x_0$ is the binary representation of $x = \sum_i 2^ix_i$.
Denote the first $l-1$ bits of $x$ as $x_{1:l} \coloneq x_1 \cdots x_{l-1}$.

\paragraph*{Martingales.} A \emph{filtration} is a sequence of $\sigma$-algebras $\set{\cF_0, \cF_1, \dots}$ satisfying $\cF_0 \subseteq \cF_1 \subseteq \cdots$.
A sequence of random variables $\set{X_k}_{k\geq 0}$ is a \emph{martingale} with respect to $\{\cF_k\}$ if $X_k$ is $\cF_k$-measurable, $\expect{\abs{X_k}} < \infty$, and $\expectcond{X_{k+1}}{\cF_k} = X_k$.
A sequence $\set{Y_k}_{k\geq 1}$ is a \emph{martingale difference sequence} (MDS) if it is adapted to the filtration and satisfies $\expectcond{Y_{k+1}}{\cF_k} = 0$.
An MDS is typically constructed from a martingale $\set{X_k}_k$ by setting $Y_k = X_k - X_{k-1}$.

\subsection{Useful lemmas}

This section collects auxiliary lemmas used in later proofs.

\begin{lemma}[Fr\'{o}chet derivative of matrix exponential, \cite{Nicholas2008functions}]~\label{lem:frochet}
    Let $A$, $E$ be two matrices. Then
\begin{equation}
    \ddx{t} e^{A + tE} = \integral{0}{1}{s}{ e^{s(A + tE)} E e^{(1-s)(A + tE)} }
.\end{equation}
\end{lemma}

\begin{lemma}~\label{lem:int norm inequality}
    For any Hermitian matrix $X$ and positive semidefinite matrix $R$,
\begin{equation}
    \tnorm{\integral{0}{1}{s}{R^s X R^{1-s}}} \leq \inorm{X} \tnorm{R}
\end{equation}
\end{lemma}
\begin{proof}
     Suppose $R$ has the spectral decomposition $R= UDU^\dagger = \sum_{j} p_j \ketbra{j}{j}$ for some orthonormal basis $\set{\ket{j}}_j$.
     Denote $\cL(X) = \integral{0}{1}{s}{D^s X D^{1-s}}$.
    Consider the $(i,j)$th element of $\cL(X)$:
\begin{align}
    \bra{i} \cL(X) \ket{j} & = \bra{i}\integral{0}{1}{s}{\sum_k p_k^s \ket{k}\bra{k} X \sum_{k'} p_{k'}^{1-s} \ket{k'}\bra{k'} }\ket{j}\\
    & = \bra{i}X\ket{j}\integral{0}{1}{s}{ p_i^s p_j^{1-s}}\\
    & = X_{i,j} L(p_i, p_j),
\end{align}
where 
\begin{equation}
L(a,b) = \begin{cases}
    a ,& \text{if}\; a=b;\\
    \frac{a-b}{\log a - \log b},& \text{if}\; a\neq b.
\end{cases}
\end{equation}
Denote the $(i,j)$th element of matrix $\mathfrak{L}$ as $L(p_i, p_j)$.
The matrix $\mathfrak{L}$ is the Loewner matrix of the Schlicht function $\log z$ mapping the upper half-plane into itself~\cite{Bhatia2007mean}, so $\mathfrak{L}$ is positive semidefinite (PSD).
Since $\cL(X)= X \odot \mathfrak{L}$, where $\odot$ is the Hadamard product, the Schur product theorem~\cite[Section 7.5]{horn2012matrix} implies that $\cL(X)$ is PSD whenever $X$ is PSD.
Therefore,
\begin{equation}
    \cL(\inorm{X} I)=\inorm{X} D\succeq\cL(X)\succeq-\inorm{X} D= \cL(-\inorm{X} I)
.\end{equation}
Denote $\lambda_1(X)$, $\lambda_2(X)$, $\dotsc$ as the eigenvalues of $X$ in ascending order.
$\cL(X)\succeq-\inorm{X} D$ means $-\lambda_j\bigl(\cL(X)\bigr)\leq \inorm{X} \lambda_j\bigl(D\bigr)$ for all $j$ satisfying $\lambda_j\bigl(\cL(X)\bigr)< 0$, while $\inorm{X} D\succeq\cL(X)$ means $\lambda_j\bigl(\cL(X)\bigr)\leq \inorm{X}\lambda_j\bigl(D\bigr)$ for all $j$ satisfying $\lambda_j\bigl(\cL(X)\bigr) \geq 0$.
Combining, we have 
\begin{equation}
\tnorm{\cL(X)}=\sum_j \abs{\lambda_j\bigl(\cL(X)\bigr)}\leq \inorm{X}\sum_j\lambda_j\bigl(D\bigr) = \inorm{X}\tnorm{D} = \inorm{X}\tnorm{R}
.\end{equation}
\end{proof}

\begin{lemma}[Baker--Campbell--Hausdorff formula for $e^Xe^Y e^X$,~\cite{hall2015lie}]~\label{lem:BCH}
    For matrices $X$ and $Y$,
\begin{equation}
    \log(e^X e^Y e^X) = 2X + Y - \frac{1}{6} \comm{X + Y}{\comm{X}{Y}} + \BigO{(\inorm{X} + \inorm{Y})^5}
.\end{equation}
\end{lemma}

\begin{lemma}[Freedman Inequality for Hermitian matrices,~\cite{tropp2011freedman}]~\label{lem:freedman}
    Let $(\cF_k)_k$ be a filtration, and let $(X_k)_k$ be a martingale difference sequence of $d$-dimensional Hermitian matrices such that $\inorm{X_k} \leq R$ for some $R > 0$.
    Then the partial sum $Y_k = \sum_{i=1}^k X_k$ and the \emph{predictable quadratic variation process} $W_k = \sum_{i=1}^k \expectcond{X_i^2}{\cF_{i-1}}$ satisfies for all $t \geq 0$ and $\sigma^2 > 0$,
\begin{equation}
    \prob{\exists k > 0 \textrm{ : } \inorm{Y_k} > t \textup{ and } \inorm{W_k} \leq \sigma^2} \leq 2d \exp\!\bigg(-\frac{t^2/2}{\sigma^2 + Rt/3}\bigg)
.\end{equation}
\end{lemma}

\begin{lemma}[Lattice-based random walk,~\cite{spitzer2001principles}]~\label{lem:reference distribution}
    Consider a random walk $(\walk{k})_k$ on a lattice $\set{-N, \ldots, N}^{\times L}$ with $\walk{0} = 0$, and there exists $p_1$, $\ldots$, $p_L$ such that for all $j$ and $k$, $ \prob{\walk{k} - \walk{k - 1} = \pm e_j} = p_j /2$.
    Then for any endpoint $\bfx \in \ZZ^{L}$ such that $\sum_j \bfx_j^2 = N$,
\begin{equation}
    \probQ{\walk{N} = \bfx} 
    = \frac{a_N(\bfx)}{(2\pi N)^{L/2}\sqrt{\prod_{j=1}^L p_j}}
    \exp\left(-\frac{1}{2N} \sum_j \frac{\bfx_j^2}{p_j}  \right)
    + o\bigl(N^{-L/2}\bigr)
,\end{equation}
    where $a_N$ is the reachable parity function in Equation~\eqref{eqn:reach parity def}.
\end{lemma}

\section{Thermal-drift channel}\label{sec:tdc}

Let $\sigma$ be an $n$-qubit Pauli operator and let $\tau \in \RR_+$. We define the unnormalized imaginary-time evolution (ITE) maps at time $\tau$ in the two opposite directions by
\begin{equation}
    \pauli{\uparrow}{\sigma}(\rho) = e^{- \tau \sigma/2}\, \rho\, e^{- \tau \sigma/2}
    \textrm{\quad and \quad }
    \pauli{\downarrow}{\sigma}(\rho) = e^{+ \tau \sigma/2}\, \rho\, e^{+ \tau \sigma/2}
.\end{equation}
Throughout this section, we use ``$\uparrow\downarrow$'' to discuss both directions simultaneously.
The goal is to construct a circuit that applies one of these two maps at random to an input state $\rho$ on the main system and records the applied direction in an ancilla register.

To formalize this task, we introduce a $2n$-qubit ancilla system $AB$ and the $n$-qubit main system $S$. Based on the identity that 
\begin{equation}
    \trace{\pauli{\uparrow}{\sigma}(\rho) + \pauli{\downarrow}{\sigma}(\rho)}
    = \trace{(e^{- \tau \sigma} + e^{\tau \sigma}) \rho} = (e^{\tau} + e^{-\tau}) \trace{\rho} = 2\cosh(\tau)
,\end{equation}
we define the thermal-drift channel $\pauliselect{\sigma}(\rho): \densityspace{S} \to \densityspace{ABS}$ satisfying
\begin{equation}~\label{eqn:pauli select}
    \trace[AB]{ \Pi^{\uparrow\downarrow}_{AB} \pauliselect{\sigma}(\rho) }
    = \frac{1}{2\cosh(\tau)} \pauli{\uparrow\downarrow}{\sigma}(\rho)
,\end{equation}
where $\Pi^{\pm}$ are mutually orthogonal projectors on $AB$ such that $\Pi^\uparrow + \Pi^\downarrow \prec I_{AB}$.
Measuring the ancilla system $AB$ in the projection-valued measure extending $\set{\Pi^\uparrow, \Pi^\downarrow}$ reveals which branch of Equation~\eqref{eqn:pauli select} occurred; the corresponding post-measurement state on $S$ is obtained by tracing out $AB$ and normalizing in the usual way.

The map satisfying Equation~\eqref{eqn:pauli select} is completely positive and trace-preserving under the chosen normalization, and thus admits a unitary dilation. The next theorem provides an explicit implementation with $\BigO{n}$ elementary gates, up to a heralded failure outcome $\Pi^\circlearrowleft$.

\begin{theorem}~\label{thm:approx pauli select}
    Suppose $\tau > 0$ and $\sigma$ is an $n$-qubit Pauli operator that has no local identity term.
    Then there exists a projection-valued measure $\set{\Pi^\uparrow, \Pi^\downarrow, \Pi^\circlearrowleft}$ on system $AB$, a $\tau$-dependent ansatz $U_\tau$ and two fixed ansatzes $T$, $V$, such that the map $\pauliselect{\approx}: \densityspace{S} \to \densityspace{ABS}$ constructed in Figure~\ref{fig:pauli select approx} (dashed area) satisfies for all $\rho \in \densityspace{S}$,
\begin{equation}
    \trace[AB]{\Pi^\circlearrowleft_{AB} \pauliselect{\approx}(\rho) } = p \rho \quad\textrm{and}\quad
    \trace[AB]{\Pi^{\uparrow\downarrow}_{AB} \pauliselect{\approx}(\rho) } = (1 - p) \trace[AB]{\Pi^{\uparrow\downarrow}_{AB} \pauliselect{\sigma}(\rho) }
,\end{equation}
    where $p^{-1} = {2\cosh^2(\tau/2)}$, $\trace{\Pi^\uparrow} = \trace{\Pi^\downarrow}$ and $\pauliselect{\sigma}$ is given as in Equation~\eqref{eqn:pauli select}. Moreover, implementation of $\pauliselect{\approx}$ takes $\BigO{n}$ single-qubit gates and CNOT gates; the projection-valued measure can be done in the computational basis.
\end{theorem}

\begin{figure}[H]
    \centering
    \includegraphics[width=\linewidth]{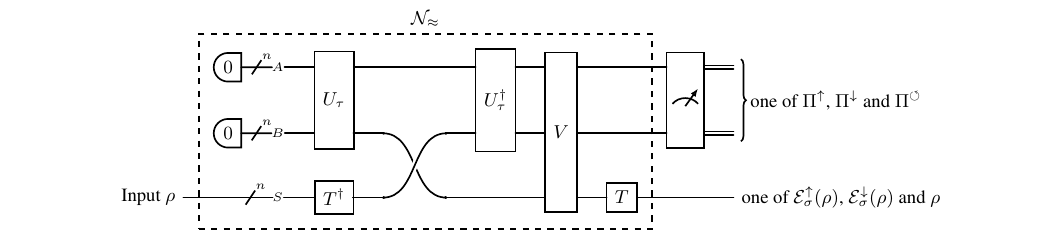}
% \begin{quantikz}[transparent]
%     {} & \inputD{0}\gategroup[3,steps=7,style={inner sep=4pt,dashed,label={above:{$\pauliselect{\approx}$}}}]{} & \qwbundle{n}_A & \gate[2]{U_\tau} & {} & \gate[2]{U_\tau^\dagger} & \gate[3]{V} & {} & \meter[2]{} & \setwiretype{c}\rstick[2]{one of $\Pi^\uparrow$, $\Pi^\downarrow$ and $\Pi^\circlearrowleft$}  \\
%     {} & \inputD{0} & \qwbundle{n}_B & {} & \gate[2,style={draw=none}]{\permute{2,1}} & {} & {} & {} & {} & \setwiretype{c} \\
%     \lstick{Input $\rho$} & {} & \qwbundle{n}_S & \gate{T^\dagger} & {} & {} & {} & \gate{T} & {} & \rstick{one of $\pauli{\uparrow}{\sigma}(\rho)$, $\pauli{\downarrow}{\sigma}(\rho)$ and $\rho$} 
% \end{quantikz}
    \caption{A circuit implementing the channel $\pauliselect{\approx}$. Repeating $\pauliselect{\approx}$ and measuring $AB$ in the computational basis until the outcome is not in $\Pi^\circlearrowleft$ implements $\pauliselect{\sigma}$.}~\label{fig:pauli select approx}
\end{figure}

Theorem~\ref{thm:approx pauli select} assumes that the Pauli operator has no local identity term; the next theorem removes this restriction and provides the resource bound for implementing a random-direction Pauli ITE step.

\begin{theorem}~\label{thm:pauli select}
    Suppose $\tau > 0$ and $\sigma$ is an $n$-qubit Pauli operator. Then there exists a quantum channel $\pauliselect{\sigma}: \densityspace{S} \to \densityspace{ABS}$ and a measure-and-trace operation $\measureselect{}: \densityspace{ABS} \to \set{1, -1} \times \densityspace{S}$ such that for all $\rho \in \densityspace{S}$,
\begin{equation}
    \measureselect{} \circ \pauliselect{\sigma}(\rho) 
    = \begin{cases}
        \left(1,\, \pauli{\uparrow}{\sigma}(\rho)/\trace{e^{-\tau\sigma}\rho}\right), & \textup{with probability } \trace{e^{-\tau\sigma} \rho} / (2\cosh \tau); \\
        \left(-1,\, \pauli{\downarrow}{\sigma}(\rho)/\trace{e^{\tau\sigma}\rho}\right), & \textup{with probability }  \trace{e^{\tau\sigma} \rho} / (2\cosh \tau).
    \end{cases}
\end{equation}
    Moreover, one can use $2n$ ancilla qubits and at most $\BigO{n \log (1/\delta)}$ single-qubit gates and CNOT gates to implement $\pauliselect{\sigma}$ with probability at least $1 - \delta$; $\measureselect{}$ requires only computational-basis measurements and classical post-processing of the outcomes.
\end{theorem}
\begin{proof}
    We first consider the case where $\sigma$ has no local identity term. Theorem~\ref{thm:approx pauli select} provides a projection-valued measure $\set{\Pi^\uparrow,\Pi^\downarrow,\Pi^\circlearrowleft}$ and a channel $\pauliselect{\approx}$ such that
\begin{equation}
     \trace[AB]{\Pi^\circlearrowleft_{AB} \pauliselect{\approx}(\rho) } = \frac{1}{2\cosh^2(\tau/2)} \rho \quad\textrm{and}\quad
    \trace[AB]{\Pi^{\uparrow\downarrow}_{AB} \pauliselect{\approx}(\rho) } = \left(1 - \frac{1}{2\cosh^2(\tau/2)}\right) \trace[AB]{\Pi^{\uparrow\downarrow}_{AB} \pauliselect{\sigma}(\rho) }
\end{equation}
    Measuring $\pauliselect{\approx}(\rho)$ on system $AB$ with respect to $\set{\Pi^\circlearrowleft, I_{AB}-\Pi^\circlearrowleft}$ partitions the output into a fallback branch to $\rho$ and a success branch to $\pauliselect{\sigma}(\rho)$.
    Define $\measureselect{\circlearrowleft}$ to be the operation that measures $AB$ and, upon obtaining $\Pi^\circlearrowleft$, traces out $AB$ and outputs the remaining state on $S$; define $\measureselect{\uparrow\downarrow}$ analogously for the complementary event $I_{AB}-\Pi^\circlearrowleft$.
    For an integer $k\geq 1$, consider the repeat-until-success map
\begin{equation}
    \measureselect{\uparrow\downarrow} \circ \pauliselect{\approx} + 
    \measureselect{\uparrow\downarrow} \circ \pauliselect{\approx} \circ \measureselect{\circlearrowleft} \circ \pauliselect{\approx} + 
    \ldots + 
    \measureselect{\uparrow\downarrow} \circ \pauliselect{\approx} \circ (\measureselect{\circlearrowleft} \circ \pauliselect{\approx})^{\circ k}
.\end{equation}
    Since the fallback probability in each invocation is at most $1/2$, the probability of fallback in all $k+1$ attempts is at most $2^{-(k+1)}$. Choosing $k=\BigO{\log(1/\delta)}$ ensures failure probability at most $\delta$.
    Each invocation of $\pauliselect{\approx}$ uses $\BigO{n}$ elementary gates by Theorem~\ref{thm:approx pauli select}, so the total gate count is $\BigO{n\log(1/\delta)}$.

    Conditioned on the success event $I_{AB}-\Pi^\circlearrowleft$, the output is proportional to $\pauliselect{\sigma}(\rho)$ as defined in Equation~\eqref{eqn:pauli select}. We now define $\measureselect{}$ to measure $AB$ with respect to $\set{\Pi^\uparrow,\Pi^\downarrow}$, trace out $AB$, and normalize the post-measurement state on $S$. 
    Using $\trace{\Pi^\uparrow}=\trace{\Pi^\downarrow}=2^{2n-2}$ and the fact that $\pauli{\uparrow}{\sigma}(\rho)$ and $\pauli{\downarrow}{\sigma}(\rho)$ have traces $\trace{e^{-\tau\sigma}\rho}$ and $\trace{e^{\tau\sigma}\rho}$, respectively, we obtain
\begin{align}
    \prob{m=1}
    &= \left(1 - \frac{1}{2\cosh^2(\tau/2)}\right)\trace{\Pi^\uparrow_{AB} \pauliselect{\sigma}(\rho)}
     = \frac{\trace{e^{-\tau\sigma}\rho}}{2\cosh\tau},\\
    \prob{m=-1}
    &= \left(1 - \frac{1}{2\cosh^2(\tau/2)}\right) \trace{\Pi^\downarrow_{AB} \pauliselect{\sigma}(\rho)}
     = \frac{\trace{e^{\tau\sigma}\rho}}{2\cosh\tau}
.\end{align}
    Conditioning on each event yields the normalized post-measurement states
    $\pauli{\uparrow}{\sigma}(\rho)/\trace{e^{-\tau\sigma}\rho}$ and
    $\pauli{\downarrow}{\sigma}(\rho)/\trace{e^{\tau\sigma}\rho}$, as stated.

    We now extend the construction to a general $\sigma$ that may contain local identity terms. Let $\sigma'$ be the tensor product of the non-identity local Pauli terms of $\sigma$, and let $P$ be an $n$-qubit permutation gate such that $\sigma = P^\dag (\sigma' \ox I) P$. For any $\rho\in \densityspace{S}$,
\begin{align}
    P^\dag(\pauli{\uparrow\downarrow}{\sigma'} \ox \idop)(P \rho P^\dag)P 
    &= P^\dag (e^{\mp \tau \sigma' /2} \ox I) P \cdot \rho \cdot P^\dag (e^{\mp \tau \sigma' /2} \ox I) P \\
    &= e^{\mp \tau \sigma/2}\, \rho\, e^{\mp \tau \sigma/2}
    = \pauli{\uparrow\downarrow}{\sigma} (\rho)
.\end{align}
    Applying the above implementation to $\sigma'$ (on the nontrivial subsystem) and conjugating by $P$ therefore realizes the required instrument for $\sigma$. The permutation $P$ can be implemented with at most $\BigO{n}$ CNOT gates, so the overall gate complexity remains $\BigO{n\log(1/\delta)}$.
\end{proof}
\vspace{1em}

The remainder of this section proves Theorem~\ref{thm:approx pauli select} by explicitly constructing $U_\tau$, $T$, and $V$.
We start from a SWAP-based construction that realizes a family of Kraus operators from a unitary on a larger space.

\begin{lemma}~\label{lem:swap enc}
    Suppose $U$ is a unitary acting on a bipartite system $AB$.
    Then for $G_{jk} = \trace[A]{U \ketbra{00}{jk} U^\dagger}$ and an additional system $B'$ such that $\abs{B'} = \abs{B}$, the following circuit of systems $ABB'$ holds
\begin{figure}[H]
    \centering
    \includegraphics[width=\linewidth]{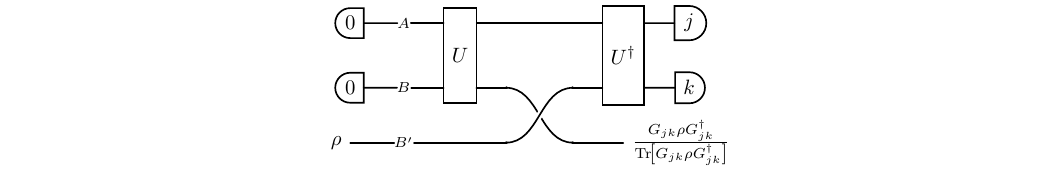}
\end{figure}
% \[
% \begin{quantikz}[transparent]
% \inputD{0} & {}_A & \gate[2]{U} & {} & \gate[2]{U^\dagger} & \meterD{j\vphantom{0}}  \\
% \inputD{0} & {}_B & {} & \gate[2,style={draw=none}]{\permute{2,1}} & {} & \meterD{k\vphantom{0}} \\
% \lstick{$\rho$} & {}_{B'} & {} & {} & \rstick{$\frac{G_{jk} \rho G_{jk}^\dag}{\trace{G_{jk} \rho G_{jk}^\dag}}$} 
% \end{quantikz}
% \]
\end{lemma}
\begin{proof}
    Let $U = \sum_{x, y, z, t} C_{xyzt} \ketbra{xy}{zt}$ and $\SWAP = \sum_{l, m} \ketbra{lm}{ml}$. For the operator acting on $B'$ conditioned on outcomes $(j,k)$ on $AB$, we compute
\begin{align}
    &\quad\,\, \left(\bra{j_A j_B} U^\dag\right)_{AB} \SWAP_{BB'} \left(U\ket{k_A k_B}\right)_{AB}  \\
    &= \sum_{x, y, z, t, m, l} C_{j_A j_Bzt}^* C_{xy k_A k_B} (\bra{z}_A \ox \bra{t}_B) (\ketbra{l}{m}_B \ox \ketbra{m}{l}_{B'}) (\ket{x}_A \ox \ket{y}_B)   \\
    &= \sum_{x, y, z, t, m, l} C_{j_A j_Bzt}^* C_{xy k_A k_B} \delta_{z,x} \delta_{t,l} \delta_{m,y} \ketbra{m}{l}_{B'}\\
    &= \sum_{x, y, z, t} C_{j_A j_B z t}^* C_{x y k_A k_B} \delta_{z,x} \ketbra{y}{t} \\
    &= \trace[A]{\sum_{x, y, z, t} C_{j_A j_Bzt}^* C_{x y k_A k_B} \ketbra{xy}{zt}_{AB} }
    = \trace[A]{U\ketbra{k_A k_B}{j_A j_B} U^\dag} 
.\end{align}
    For any pure state $\ket{\psi}$ on $B'$, substituting this identity into the circuit shows that
\begin{align}
    U^\dag_{AB} \SWAP_{BB'} U_{AB} (\ket{00} \ox \ket{\psi})
    &= \sum_{j, k} \ketbra{jk}{jk}_{AB} U^\dag_{AB} \SWAP_{BB'} U_{AB} (\ket{00} \ox \ket{\psi}) \\
    &= \sum_{j, k} \ket{jk}_{AB} \left(\bra{jk} U^\dag\right)_{AB} \SWAP_{BB'} \left(U\ket{00}\right)_{AB} \ket{\psi}_{B'} \\
    &= \sum_{j, k} \ket{jk}_{AB} G_{jk} \ket{\psi}_{B'} 
    = \sum_{j, k} \ket{jk} \ox G_{jk} \ket{\psi}
.\end{align}
    When the system $AB$ is measured to be $jk$, the unnormalized post-measurement state on $B'$ is $G_{jk}\ket{\psi}$, and hence the normalized post-measurement state is $G_{jk}\ket{\psi}/\norm{G_{jk}\ket{\psi}}$.
    By linearity, the same conclusion holds for a mixed input $\rho$, yielding the claimed output state $G_{jk}\rho G_{jk}^\dag/\trace{G_{jk}\rho G_{jk}^\dag}$.
\end{proof}
\vspace{1em}

We now prove Theorem~\ref{thm:approx pauli select}. All three systems $A, B, S$ are $n$-qubit. In the following analysis, $A_i$ denotes the $i$-th qubit in system $A$ (counting from top to bottom), and similarly for $B$ and $S$. 

\renewcommand\thetheorem{\ref{thm:approx pauli select}}
\begin{theorem}[Complete version]
    Suppose $\tau > 0$ and $\sigma$ is an $n$-qubit Pauli operator that has no local identity term.
    Consider the following construction of a projection-valued measure on the system $AB$,
\begin{equation}~\label{eqn:proj select}
    \Pi^\circlearrowleft = \ketbra{1}{1}_{A_n},\quad
    \Pi^\uparrow = \ketbra{0}{0}_{A_n} \ox \ketbra{0}{0}_{B_n},\quad
    \Pi^\downarrow = \ketbra{0}{0}_{A_n} \ox \ketbra{1}{1}_{B_n}
,\end{equation}
    and of circuit ansatzes on system $ABS$,
\begin{align}
    U_\tau &= \left(\prod_{i=1}^n\CNOT_{A_i B_i}\right) \left(\prod_{i=1}^{n-1}\CNOT_{A_i A_n}\right) \left(H^{\ox (n-1)} \ox R_y(2\theta)\right)_A, \\
    V &= \left(\prod_{i=1}^{n} \CNOT_{B_i B_n} \right)\left( \prod_{i=1}^n \CNOT_{B_i S_i} \CZ_{A_i S_i} \right)\left(\prod_{i=1}^{n-1} \CNOT_{A_n A_i} \right), \\
    T & \textrm{\,\, satisfies\quad} TZ^{\ox n} T^\dag = \sigma
,\end{align}
    where $\theta = \arccos\sqrt{e^{-\tau/2} / \bigl(2\cosh(\tau/2)\bigr)}$.  
    Then the map $\pauliselect{\approx}$ constructed in Figure~\ref{fig:pauli select approx} (dashed area) satisfies 
\begin{equation}
    \trace[AB]{\Pi^\circlearrowleft_{AB} \pauliselect{\approx}(\rho) } = \frac{1}{2\cosh^2(\tau/2)} \rho \quad\textrm{and}\quad
    \trace[AB]{\Pi^{\uparrow\downarrow}_{AB} \pauliselect{\approx}(\rho) } = \frac{1}{4\cosh^2(\tau/2)} \pauli{\uparrow\downarrow}{\sigma}(\rho)
.\end{equation}
\end{theorem}
\renewcommand{\thetheorem}{S\arabic{proposition}}
\begin{proof}
    The proof is decomposed into two parts: we first \textbf{(1)} analyze what $G_{jk}$ is by taking $U_\tau$ into Lemma~\ref{lem:swap enc}, and show that \textbf{(2)} $V$ can decode, map and group these  $G_{jk}$ to be either $e^{-\tau \sigma/2}$, $e^{\tau \sigma/2}$ or the identity matrix. Taking the whole picture gives the statement.

\textbf{(1)} For the first part, one can derive that for all $\ket{j} = \ket{j_1\cdots j_n}$ in system $A$,
\begin{align}
    &\quad\,\, \left(\prod_{i=1}^{n-1}\CNOT_{A_i A_n}\right) \left(H^{\ox (n-1)} \ox R_y(2\theta)\right) \ket{j} \\
    &= \left(\prod_{i=1}^{n-1}\CNOT_{A_i A_n}\right) \left(H^{\ox (n-1)} \ket{j_{1:n}} \ox R_y(2\theta) \ket{j_n}\right) \\
    &= \frac{1}{\sqrt{2^{n - 1}}} \left(\prod_{i=1}^{n-1}\CNOT_{A_i A_n}\right) \sum_{x \in \set{0, 1}^{n - 1}} (-1)^{x \cdot j_{1:n}} \ket{x} \ox \left(\cos (\theta) \ket{j_n} + (-1)^{j_n} \sin(\theta) X \ket{j_n} \right) \\
    &= \frac{1}{\sqrt{2^{n - 1}}} \sum_{x \in \set{0, 1}^{n - 1}} (-1)^{x \cdot j_{1:n}} \ket{x} \ox \Big(\cos (\theta) \ket{j_n \oplus \parity{x}} + (-1)^{j_n} \sin(\theta) \ket{j_n \oplus \parity{x} \oplus 1} \Big)
,\end{align}
    where $j_{1:n} = j_1 \cdots j_{n-1}$, $\oplus$ is the bitwise addition and $\parity{x}$ is the bitwise parity function defined in Equation~\eqref{eqn:parity def}. Taking the last piece of $U_\tau$ gives
\begin{align}
    U_\tau \ket{jk} 
    &= \left(\prod_{i=1}^n\CNOT_{A_i B_i}\right) \left(\prod_{i=1}^{n-1}\CNOT_{A_i A_n}\right) \left(H^{\ox (n-1)} \ox R_y(2\theta)\right) (\ket{j}_A \ox \ket{k}_B) \\
    &= \left(\prod_{i=1}^n\CNOT_{A_i B_i}\right) \frac{1}{\sqrt{2^{n - 1}}} \sum_{x \in \set{0, 1}^{n - 1}} (-1)^{x \cdot j_{1:n}} \ket{x}_{A_{1:n}} \ox \Big(\cos (\theta) \ket{j_n'}_{A_n} +  (-1)^{j_n} \sin(\theta) \ket{j_n' \oplus 1}_{A_n} \Big) \ox \ket{k}_B  \\
    &= \sum_{x \in \set{0, 1}^{n - 1}}  \frac{(-1)^{x \cdot j_{1:n}}}{\sqrt{2^{n - 1}}} \ket{x}_{A_{1:n}} \ket{k_{1:n} \oplus x}_{B_{1:n}} \Big(\cos (\theta) \ket{j_n'}_{A_n} \ket{k_n'}_{B_n} + (-1)^{j_n} \sin(\theta) \ket{j_n' \oplus 1}_{A_n} \ket{k_n' \oplus 1}_{B_n} \Big)
\end{align}
    for $j_n' = j_n \oplus \parity{x}$ and $k_n' = k_n \oplus j_n'$. When $j$ and $k$ are both zero, we have
\begin{equation}
    U_\tau \ket{00}
    = \frac{1}{\sqrt{2^{n - 1}}} \sum_{x \in \set{0, 1}^{n - 1}} \ket{x}_{A_{1:n}} \ket{x}_{B_{1:n}} \Big(\cos (\theta) \ket{\parity{x}}_{A_n} \ket{\parity{x}}_{B_n} + \sin(\theta) \ket{\parity{x} \oplus 1}_{A_n} \ket{\parity{x} \oplus 1}_{B_n} \Big)
.\end{equation}
    Then we have
\begin{align}
    U_\tau \ketbra{00}{jk} U_\tau^\dag 
    = \frac{1}{2^{n - 1}} \sum_{x \in \set{0, 1}^{n - 1}} &(-1)^{x \cdot j_{1:n}}  \ketbra{x}{x}_{A_{1:n}} \ketbra{x}{k_{1:n} \oplus x}_{B_{1:n}} \cdot \\
    & \Big(\cos (\theta) \ket{\parity{x}}_{A_n} \ket{\parity{x}}_{B_n} + \sin(\theta) \ket{\parity{x} \oplus 1}_{A_n} \ket{\parity{x} \oplus 1}_{B_n} \Big) \cdot \\
    & \Big(\cos (\theta) \bra{j_n'}_{A_n} \bra{k_n'}_{B_n} + (-1)^{j_n} \sin(\theta) \bra{j_n' \oplus 1}_{A_n} \bra{k_n' \oplus 1}_{B_n} \Big).
\end{align}
    $G_{jk}$ in Lemma~\ref{lem:swap enc} is given as
\begin{equation}
    G_{jk} 
    = \trace[A]{U_\tau \ketbra{00}{jk} U_\tau^\dag} 
    = \frac{1}{2^{n - 1}} \sum_{x \in \set{0, 1}^{n - 1}} (-1)^{x \cdot j_{1:n}} \ketbra{x}{k_{1:n} \oplus x} \ox Q_{j_n k_n}
,\end{equation}
    where $Q_{j_n k_n}$ is constructed as
\begin{align}~\label{eqn:Q construction}
    Q_{j_n k_n} &= \delta_{\parity{x},j_n'} \left( \cos^2(\theta) \ketbra{\parity{x}}{k_n'} + (-1)^{j_n} \sin^2(\theta) X \ketbra{\parity{x}}{k_n'} X \right) + \\
    &\quad\frac{1}{2} \delta_{\parity{x},(j_n' \oplus 1)} \sin(2\theta) \left( X \ketbra{\parity{x}}{k_n'} + (-1)^{j_n} \ketbra{\parity{x}}{k_n'} X \right)
.\end{align}

\textbf{(2)} For the second part, without loss of generality, suppose $\rho = \ketbra{\psi}{\psi}$ is a pure state. 
At this point, by Lemma~\ref{lem:swap enc}, the output state $\ket{\psi'}$ of the following circuit
\begin{figure}[H]
    \centering
    \includegraphics[width=\linewidth]{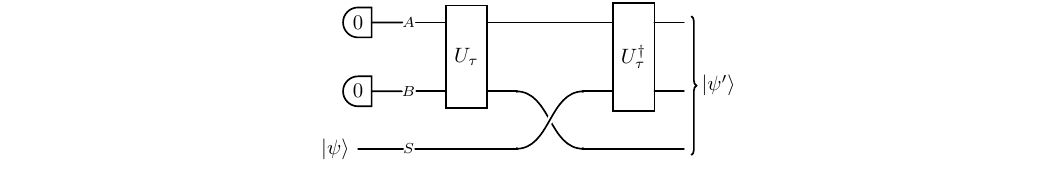}
\end{figure}
% \[
% \begin{quantikz}[transparent]
% \inputD{0} & {}_A & \gate[2]{U_\tau} & {} & \gate[2]{U_\tau^\dagger} & \rstick[3]{$\ket{\psi'}$}   \\
% \inputD{0} & {}_B & {} & \gate[2,style={draw=none}]{\permute{2,1}} & {}& {} \\
% \lstick{$\ket{\psi}$} & {}_{S} & {} & {} & {} & {} 
% \end{quantikz}
% \]
can be expressed as $\ket{\psi'} = \sum_{j, k} \ket{jk} \ox G_{jk} \ket{\psi}$. Then
\begin{align}
    V\ket{\psi'} 
    &= \left(\prod_{i=1}^{n} \CNOT_{B_i B_n} \right) \left( \prod_{i=1}^n \CNOT_{B_i S_i} \CZ_{A_i S_i} \right) \left(\prod_{i=1}^{n-1} \CNOT_{A_n A_i} \right) \sum_{j, k} \ket{j} \ox \ket{k} \ox G_{jk} \ket{\psi} \\
    &= \left(\prod_{i=1}^{n} \CNOT_{B_i B_n} \right) \left( \prod_{i=1}^n \CNOT_{B_i S_i} \CZ_{A_i S_i} \right) \sum_{j, k} \left( \bigotimes_{i=1}^{n-1}\ket{j_i \oplus j_n} \right) \ox \ket{j_n} \ox \ket{k} \ox G_{jk} \ket{\psi} \\
    &= \left(\prod_{i=1}^{n} \CNOT_{B_i B_n} \right) \left( \prod_{i=1}^n \CNOT_{B_i S_i} \right) \left( \prod_{i=1}^n \CZ_{A_i S_i} \right) \sum_{j, k} \left( \bigotimes_{i=1}^{n-1}\ket{j_i \oplus j_n} \right) \ox \ket{j_n} \ox \ket{k} \ox G_{jk} \ket{\psi} \\
    &= \left(\prod_{i=1}^{n} \CNOT_{B_i B_n} \right) \sum_{j, k} \left( \bigotimes_{i=1}^{n-1}\ket{j_i \oplus j_n} \right) \ox \ket{j_n} \ox \ket{k} \ox E_{jk} \ket{\psi} \\
    &= \sum_{j, k} \left( \bigotimes_{i=1}^{n-1}\ket{j_i \oplus j_n} \right) \ox \ket{j_n} \ox \ket{k_{1:n}} \ox \ket{\parity{k}} \ox E_{jk} \ket{\psi}
\end{align}
    for $E_{jk}$ constructed as
\begin{equation}
    E_{jk} = X^{(k)} \left(\prod_{i = 1}^{n-1} Z^{j_i + j_n}_{S_i}\right) Z^{j_n}_{S_n} G_{jk} 
,\end{equation}
    where $X^{(k)}$ is short for $\bigotimes_{i=1}^n X^{k_i}$.
    Now we analyze what $E_{jk}$ will be for three cases: $j_n = 1$, $(j_n, \parity{k}) = (0, 0)$ and $(j_n, \parity{k}) = (0, 1)$.

    If $j_n = 1$, then $j_n' = j_n \oplus \parity{x} = \parity{x} \oplus 1$, $k_n' = k_n \oplus j_n' = k_n \oplus \parity{x} \oplus 1$ and $Q_{j_n k_n}$ in Equation~\eqref{eqn:Q construction} satisfies
\begin{equation}
    Q_{j_n k_n} = \frac{1}{2} \sin(2\theta) \left( X\ketbra{\parity{x}}{k_n \oplus \parity{x}} X - \ketbra{\parity{x}}{k_n \oplus \parity{x}} \right)
\end{equation}
    and hence
\begin{equation}
    G_{jk} 
    = \frac{1}{2^n} \sin(2\theta) \sum_{x \in \set{0, 1}^{n - 1}} (-1)^{x \cdot j_{1:n}} \ketbra{x}{k_{1:n} \oplus x} \ox  \left( X\ketbra{\parity{x}}{k_n \oplus \parity{x}} X - \ketbra{\parity{x}}{k_n \oplus \parity{x}} \right)
\end{equation}
    and hence
\begin{align}
    E_{jk}
    &= \frac{\sin(2\theta)}{2^n} X^{(k)} \left(\prod_{i = 1}^{n-1} Z^{j_i +1}_{S_i}\right) Z_{S_n} \sum_{x \in \set{0, 1}^{n - 1}} (-1)^{x \cdot j_{1:n}} \ketbra{x}{k_{1:n} \oplus x} \ox  \left( X\ketbra{\parity{x}}{k_n \oplus \parity{x}} X - \ketbra{\parity{x}}{k_n \oplus \parity{x}} \right) \\
    &= \frac{\sin(2\theta)}{2^n} X^{(k)} \sum_{x \in \set{0, 1}^{n - 1}} (-1)^{x} \ketbra{x}{k_{1:n} \oplus x} \ox \left( (-1)^{\parity{x} + 1} \ketbra{\parity{x} \oplus 1}{k_n \oplus \parity{x} \oplus 1} - (-1)^{\parity{x}} \ketbra{\parity{x}}{k_n \oplus \parity{x}} \right) \\
    &= \frac{\sin(2\theta)}{2^n} \sum_{x \in \set{0, 1}^{n - 1}} (-1)^{x} \ketbra{x \oplus k_{1:n}}{k_{1:n} \oplus x} \ox (-1)^{\parity{x} + 1} \left( \ketbra{\parity{x} \oplus k_n \oplus 1}{k_n \oplus \parity{x} \oplus 1} + \ketbra{\parity{x} \oplus k_n}{k_n \oplus \parity{x}} \right) \\
    &= - \frac{\sin(2\theta)}{2^n} \sum_{x \in \set{0, 1}^{n - 1}} \ketbra{x \oplus k_{1:n}}{k_{1:n} \oplus x} \ox I 
    = - \frac{1}{2^n} \sin(2\theta) I_S
.\end{align}

    If $j_n = 0$, then $j_n' = j_n \oplus \parity{x} = \parity{x}$, $k_n' = k_n \oplus j_n' = k_n \oplus \parity{x}$ and hence
\begin{align}
    Q_{j_n k_n}
    &= \cos^2(\theta) \ketbra{\parity{x}}{k_n \oplus \parity{x}} + \sin^2(\theta) X \ketbra{\parity{x}}{k_n \oplus \parity{x}} X \\
    \implies G_{jk} 
    &= \frac{1}{2^{n - 1}} \sum_{x \in \set{0, 1}^{n - 1}} (-1)^{x \cdot j_{1:n}} \ketbra{x}{k_{1:n} \oplus x} \ox \left( \cos^2(\theta) \ketbra{\parity{x}}{k_n \oplus \parity{x}} +  \sin^2(\theta) X \ketbra{\parity{x}}{k_n \oplus \parity{x}} X \right)
,\end{align}
    which gives
\begin{align}
    E_{jk} 
    &= X^{(k)} Z^{(j_{1:n})}_{S_{1:n}} G_{jk} \\
    &= \frac{1}{2^{n - 1}} X^{(k)} \sum_{x \in \set{0, 1}^{n - 1}} \ketbra{x}{k_{1:n} \oplus x} \ox \left( \cos^2(\theta) \ketbra{\parity{x}}{k_n \oplus \parity{x}} + \sin^2(\theta) X \ketbra{\parity{x}}{k_n \oplus \parity{x}} X \right) \\
    &= \frac{1}{2^{n - 1}} X^{(k)} \left[ \sum_{x \in \set{0, 1}^{n - 1}} \ketbra{x}{x} \ox \left( \cos^2(\theta) \ketbra{\parity{x}}{\parity{x}} + \sin^2(\theta) X \ketbra{\parity{x}}{\parity{x}} X \right) \right] X^{(k)}
.\end{align}
    Denote $C = 2\cosh(\tau/2) = e^{-\tau/2} + e^{\tau/2} $. Since $\cos^2(\theta) = e^{-\tau/2} / C$ and $\sin^2(\theta) = e^{\tau/2} / C$, the expression surrounded by the square bracket is simplified as
\begin{align}
    &\quad\,\, \sum_{x \in \set{0, 1}^{n - 1}} \ketbra{x}{x} \ox \left( \cos^2(\theta) \ketbra{\parity{x}}{\parity{x}} + \sin^2(\theta) X \ketbra{\parity{x}}{\parity{x}} X \right) \\
    &= \frac{1}{C} \sum_{x \in \set{0, 1}^{n - 1}} \ketbra{x}{x} \ox \left( e^{-\tau/2} \ketbra{\parity{x}}{\parity{x}} + e^{\tau/2} \ketbra{\parity{x} \oplus 1}{\parity{x} \oplus 1} \right)\\
    &= \frac{1}{C} \sum_{x \in \set{0, 1}^{n - 1}} \ketbra{x}{x} \ox \exp(-\frac{\tau}{2} (-1)^{\parity{x}} Z)
    =\frac{1}{C} e^{-\tau Z^{\ox n}/2}
,\end{align}
    where the last equality comes from the fact that $Z^{\ox (n - 1)} \ket{x} = (-1)^{\parity{x}} \ket{x}$. Then
\begin{align}
    E_{jk} 
    &= \frac{1}{2^{n - 1}C} X^{(k)} e^{-\tau Z^{\ox n}/2} X^{(k)} \\
    &= \frac{1}{2^{n - 1}C} \begin{cases}
        e^{-\tau Z^{\ox n}/2}, &\textrm{if } \parity{k} = 0, \\
        e^{\tau Z^{\ox n}/2}, &\textrm{if } \parity{k} = 1.
    \end{cases}
\end{align}
    As a result, omitting subsystems $A_{1:n}$ and $B_{1:n}$, since $T Z^{\ox n} T^\dag = \sigma$, we have
\begin{align}
    T_S V T_S^\dag \ket{\psi'} 
    =& \sum_{j, k} \ldots \ox \ket{j_n}_{A_n} \ox \ket{\parity{k}}_{B_n} \ox T E_{jk} T^\dag \ket{\psi} \\
    =& -\sum_{j_n = 1} \ldots \ox \ket{1}_{A_n} \ox \ket{\parity{k}}_{B_n} \ox \frac{1}{2^n} \sin(2\theta) \ket{\psi} + \\
    & \sum_{(j_n, \parity{k}) = (0, 0)} \ldots \ox \ket{0}_{A_n} \ox \ket{0}_{B_n} \ox \frac{1}{2^{n - 1}C} e^{-\tau \sigma/2}\ket{\psi} + \\
    & \sum_{(j_n, \parity{k}) = (0, 1)} \ldots \ox \ket{0}_{A_n} \ox \ket{1}_{B_n} \ox \frac{1}{2^{n - 1}C} e^{\tau \sigma/2}\ket{\psi}
.\end{align}
    There are $2^{2n-1}$, $2^{2n-2}$, $2^{2n-2}$ pairs of $(j, k)$ for satisfying $j_n = 1$, $(j_n, \parity{k}) = (0, 0)$ and $(j_n, \parity{k}) = (0, 1)$, respectively. Finally, taking $\Pi^\circlearrowleft$, $\Pi^\uparrow$, $\Pi^\downarrow$ in Equation~\eqref{eqn:proj select} gives
\begin{align}
    \trace[AB]{\Pi^\circlearrowleft_{AB} \pauliselect{\approx}(\ketbra{\psi}{\psi}) }
    &= 2^{2n-1} \cdot \frac{1}{2^{2n}} \sin^2(2\theta) \ketbra{\psi}{\psi} = \frac{1}{2\cosh^2(\tau/2)} \ketbra{\psi}{\psi},  \\
    \trace[AB]{\Pi^\uparrow_{AB} \pauliselect{\approx}(\ketbra{\psi}{\psi}) } 
    &= 2^{2n-2} \cdot \frac{1}{2^{2n - 2}C^2} e^{-\tau \sigma/2}\ketbra{\psi}{\psi}e^{-\tau \sigma/2} = \frac{1}{4\cosh^2(\tau/2)} \pauli{\uparrow}{\sigma}(\ketbra{\psi}{\psi}), \\
    \trace[AB]{\Pi^\downarrow_{AB} \pauliselect{\approx}(\ketbra{\psi}{\psi}) }
    &= 2^{2n-2} \cdot \frac{1}{2^{2n - 2}C^2} e^{\tau \sigma/2}\ketbra{\psi}{\psi}e^{\tau \sigma/2} = \frac{1}{4\cosh^2(\tau/2)} \pauli{\downarrow}{\sigma}(\ketbra{\psi}{\psi})
.\end{align}
    The statement holds by generalizing $\ketbra{\psi}{\psi}$ to $\rho$ by linearity.
\end{proof}

\section{Sampling thermal states and their labels}~\label{sec:sample proof}

Let $\beta > 0$ and let $\Sigma = \set{\sigma_j}_{j=1}^L$ be a set of $L$ $n$-qubit Pauli operators with a boundary vector ${\bmh} \in \RR_{+}^{L}$.
The goal of random thermal state generation is to sample a datum $(\gibbs{\beta}{H}, H)$ from the dataset
\begin{equation}
\dataset = \setcond{ \left(\gibbs{\beta}{H}, H\right) }{ H = \sum_{j=1}^L c_j \sigma_j,\; \abs{c_j} \leq \bmh_j,\; c_j \in \RR,\; \sigma_j \in \Sigma}
\end{equation}
where $\gibbs{\beta}{H} = e^{-\beta H} / \trace{e^{-\beta H}}$ is the thermal state of the Hamiltonian $H$ at temperature $\beta^{-1}$.
To avoid sampling continuous coefficients directly, we sample from a quantized surrogate. For $N \geq 1$ define
\begin{equation}
    \datasetN = \setcond{ \left(\gibbs{\beta}{H}, H\right) }{ H = \sum_{j=1}^L c_j \sigma_j,\; \abs{c_j} \leq \bmh_j,\; c_j/\bmh_j \in (1/N)\ZZ,\; \sigma_j \in \Sigma}
\end{equation}
so that $\bigcup_{N\geq 1} \datasetN$ is dense in $S(\beta,\Sigma,\bmh)$ with respect to the sup-norm on the coefficient vector. 
The following algorithm describes how to sample a data from $\datasetN$.

\begin{algorithm}[H]\label{alg:sample thermal}
\caption{Pseudo code for our thermal state sampler}
\SetKwInOut{Input}{Input}
\SetKwInOut{Output}{Output}
\Input{Inverse temperature $\beta$, a set $\Sigma$ of $n$-qubit Pauli operators with a boundary vector ${\bmh}$, a number of steps $N$}
\Output{A datum $(\gibbs{\beta}{H}, H) \in \datasetN$}

$\tau \gets \hsum \beta / N$, $\rho \gets I^{\ox n}/2^n$\;

$\mathrm{SAMPLE} \gets $ a classical function that returns a value $j$ with weight $p_j = \bmh_j / \hsum$\;

\For{$k = 1, \ldots, N$}{
    $j_k \gets \mathrm{SAMPLE}()$\;

    $\measureselect{j_k}, \pauliselect{\sigma_{j_k}} \gets $ Theorem~\ref{thm:pauli select} wrt. $\sigma_{j_k}$.

    $(m_k, \rho)  \gets (\measureselect{j_k} \circ \pauliselect{\sigma_{j_k}})(\rho)$\;
}

$c_j \gets \hsum \sum_{k:\, j_k = j}  m_k / N$, $H \gets \sum_{j=1}^L c_j \sigma_j$\;

\Return{$(H, \rho)$}
\end{algorithm}

This section develops two technical components used to analyze Algorithm~\ref{alg:sample thermal}. First, we prove the trace-norm error bound between the algorithmic output state and the target thermal state (Theorem~\ref{thm:error appendix}). Second, we set up the random-walk viewpoint for the induced Hamiltonian coefficients, which underlies the sampling-probability analysis in Theorem~\ref{thm:distribution appendix}.

\subsection{Error analysis}~\label{sec:error}

Fix the step number $N$ and a sequence of indices $(j_k)_{k=1}^N \subset [L]$ selecting Pauli operators from $\Sigma = \set{\sigma_j}_{j=1}^L$, as in Algorithm~\ref{alg:sample thermal}. Let $\rho_k$ denote the (normalized) system state after $k$ steps, and define the step size $\tau = \hsum \beta /N$. Then $(\rho_k)_{k=0}^N$ evolves according to the Markovian update rule
\begin{equation}~\label{eqn:alg step rho}
    \rho_0 = I/2^n,
    \quad
    \rho_k 
    = \measureselect{j_k} \circ \pauliselect{j_k} (\rho_{k-1}) = \frac{e^{B_k} \,\rho_{k-1}\, e^{B_k}}{\trace{e^{B_k} \,\rho_{k-1}\, e^{B_k}}}
,\end{equation}
where $\measureselect{j_k}, \pauliselect{j_k}$ are specified by Theorem~\ref{thm:pauli select} for $\sigma_{j_k}$. The corresponding random variable $B_k$ takes the form
\begin{equation}
    B_k =\begin{cases}
        \tau\sigma_{j_k}/2, &\textrm{ with probability } {\trace{e^{\tau \sigma_{j_k}} \rho_{k-1}}}/{2\cosh\tau}, \\
        -\tau\sigma_{j_k}/2, &\textrm{ with probability } {\trace{e^{-\tau \sigma_{j_k}} \rho_{k-1}}}/{2\cosh\tau}.
    \end{cases} 
\end{equation}
Iterating Equation~\eqref{eqn:alg step rho} shows that the final output state can be written as
\begin{equation}
    \rho_N 
    = \frac{\big(\prod_{k=N}^1 e^{B_k}\big) \rho_{0} \big(\prod_{k=1}^N e^{B_k}\big)}{\trace{\text{ same above}}} 
    = \frac{\big(\prod_{k=N}^2 e^{B_k}\big) e^{2B_1} \big(\prod_{k=2}^N e^{B_k}\big)}{\trace{\text{ same above}}}
,\end{equation}
where $\prod_{k=N}^2$ denotes the ordered product from $N$ down to $2$. The corresponding target thermal state is
\begin{equation}~\label{eqn:expect output}
    \gibbs{\beta}{H} = \frac{e^{-\beta H}}{\trace{e^{-\beta H}}}
    \textrm{\quad for \quad}
    H = -\frac{2}{\beta}\sum_{k=1}^N B_k
.\end{equation}
In this subsection, we prove Theorem~\ref{thm:error appendix}, which shows that, with probability at least $1-\delta$ and for sufficiently large $N$,
\begin{equation}
    \tnorm{\rho_N - \gibbs{\beta}{H}} \leq \BigO{\big(n + \log(N/\delta)\big)^{3/2} \hsum^{3} \beta^3 N^{-3/2} }
.\end{equation}
Two lemmas simplify the proof. The first lemma compares the palindromic product in $\rho_N$ with the exponential of the accumulated generator.

\begin{lemma}~\label{lem:BCH multi expansion}
   Let $(B_k)_{k=1}^N$ be a sequence of matrices. Denote $S_k = \sum_{j=1}^k B_j$. Suppose $\max_{1\leq j\leq N}\inorm{B_j}=\tau\ll 1$. Then
\begin{equation}
    \log\left(\big(\prod_{j=N}^2 e^{B_j}\big) e^{2B_1} \big(\prod_{j=2}^N e^{B_j}\big)\right)
    = 2S_N +\Delta_N + \cR_N 
,\end{equation}
    where $\Delta_k = -\frac{1}{6}\sum_{j=2}^k \comm{B_j + 2S_{j-1}}{\comm{B_j}{2S_{j-1}}}$ denotes the terms of order $3$ and $\cR_N$ denotes the remainder terms.
    We have
    \begin{equation}
         \inorm{\cR_{N}} =\BigO{\sum_{r=5, r \textrm{ is odd}}^\infty \frac{1}{r!}\sum_{k=2}^{N} \inorm{\operatorname{ad}^{r-1}_{2S_{k-1}}(B_k)}}.
    \end{equation}
\end{lemma}

\begin{proof}
For Hermitian $X$ and $Y$, since $(e^X e^Y e^X)^{-1} = e^{-X} e^{-Y} e^{-X}$, we have
\begin{align}
    \log \left(e^X e^Y e^X \right) + \log \left(e^{-X} e^{-Y} e^{-X} \right)=0,
\end{align}
which implies $f(X,Y)=\log \left(e^X e^Y e^X \right)=-f(-X,-Y)$.
Hence, the BCH expansion of $\log \left(e^X e^Y e^X \right)$ contains no even-order terms.

We proceed by induction using the palindromic products
\begin{equation}
    M_k \coloneqq \Big(\prod_{j=k}^2 e^{B_j}\Big) e^{2B_1} \Big(\prod_{j=2}^k e^{B_j}\Big),
    \quad
    Z_k \coloneqq \log(M_k)
.\end{equation}
We have $M_1 = e^{2B_1}$, and for $k\geq 2$, $M_k = e^{B_k} M_{k-1} e^{B_k} = e^{B_k} e^{Z_{k-1}} e^{B_k}$.

For $k = 1$, we have $Z_1 = \log(e^{2B_1})= 2B_1= 2S_1$, so the claim holds.

Now fix $k \geq 2$ and assume inductively that $Z_{k-1} = 2S_{k-1} + \Delta_{k-1} + \cR_{k-1}$, where $\Delta_{k-1}$ collects the cubic terms and $\cR_{k-1}$ collects the remaining odd-order terms of degree at least $5$. Applying Lemma~\ref{lem:BCH} to $e^{B_k} e^{Z_{k-1}} e^{B_k}$ gives
\begin{align}
    Z_k & = \log \left(e^{B_k} e^{Z_{k-1}} e^{B_k}\right) \\
    &=\log \left(e^{B_k} e^{2S_{k-1} + \Delta_{k-1} + \cR_{k-1}} e^{B_k}\right)\\
    & = 2S_k + \Delta_{k-1}+ \cR_{k-1} -\frac{1}{6} \comm{B_k+2S_{k-1} + \Delta_{k-1} + \cR_{k-1}}{\comm{B_k}{2S_{k-1} + \Delta_{k-1} + \cR_{k-1}}}.
\end{align}
Separating the cubic contribution from higher-order terms yields
\begin{align}
    \Delta_k &= \Delta_{k-1}-\frac{1}{6} \comm{B_k+2S_{k-1}}{\comm{B_k}{2S_{k-1}}},\\
    \cR_{k} & = \cR_{k-1} -\frac{1}{6}\comm{B_k+2S_{k-1}}{\comm{B_k}{\Delta_{k-1} + \cR_{k-1}}} - \frac{1}{6}\comm{\Delta_{k-1} + \cR_{k-1}}{\comm{B_k}{2S_{k-1} + \Delta_{k-1} + \cR_{k-1}}}.
\end{align}
Iterating the recursion for $\Delta_k$ gives
\begin{align}
    \Delta_N &=\Delta_{N-1} -\frac{1}{6} \comm{B_N+2S_{N-1}}{\comm{B_N}{2S_{N-1}}}\\
    &=-\frac{1}{6}\sum_{k=2}^{N}\comm{B_k+2S_{k-1}}{\comm{B_k}{2S_{k-1}} }.
\end{align}

It remains to bound the remainder terms in $\cR_N$. Write $\cR_k=\sum_{r\geq 5,\, r \text{ odd}} \cR_{k,r}$, where $\cR_{k,r}$ denotes the homogeneous terms of total order $r$, and assume $\inorm{B_{k}}=\BigO{\inorm{S_{k-1}}}$ for $k\geq 2$.
From the recursive definition of $\cR_k$, the order-$r$ component obeys
\begin{align}
    \cR_{k, r} = \cR_{k-1, r} +\frac{1}{6}\comm{B_k + 2S_{k-1}}{\comm{\cR_{k-1, r-2}}{B_k}} +\frac{1}{6} \sum_{j=3}^{r-4}\comm{\cR_{k-1, j}}{\comm{\cR_{k-1, r-1-j}}{B_k}}.
\end{align}
We claim that, for each odd $r\geq 5$, the increment $\cR_{k,r}-\cR_{k-1,r}$ is dominated (up to constants depending only on $r$) by terms of the form $\operatorname{ad}^{r-1}_{2S_{k-1}}(B_k)$. The claim is immediate for $r=5$, and the recursion preserves the dominance for higher odd orders.

Finally, since the coefficient of $\operatorname{ad}^{r-1}_X(Y)$ in the BCH expansion of $\log (e^X e^Y e^X)$ has magnitude $2/r!$, there exists a constant $C_r$ (depending only on $r$) such that, for all odd $r\geq 5$,
\begin{align}
    \inorm{\cR_{N, r}} &\leq \frac{C_r}{r!} \sum_{k=2}^{N} \inorm{\operatorname{ad}^{r-1}_{2S_{k-1}}(B_k)}.
\end{align}
Summing over odd $r\geq 5$ gives the stated bound on $\inorm{\cR_N}$.
\end{proof}
\vspace{1em}

The second lemma controls the effect of normalizing matrix exponentials.

\begin{lemma}~\label{lem:error bound normalize}
    Let $A$, $E$ be two Hermitian matrices. Then
\begin{equation}
    \tnorm{\frac{e^{A + E}}{\trace{e^{A + E}}} - \frac{e^A}{\trace{e^A}}} \leq 2 \inorm{E}
.\end{equation}
\end{lemma}
\begin{proof}
    Consider $t \in [0, 1]$, $A(t) = A + tE$ and $\rho(t) = e^{A(t)}/\cZ(t)$ with $\cZ(t) = \trace{e^{A(t)}}$.
    By Lemma~\ref{lem:frochet}, we have
\begin{equation}~\label{eqn:drhot-dt}
    \ddx{t} e^{A(t)} = \integral{0}{1}{s}{ e^{sA(t)} E e^{(1-s)A(t)} }
.\end{equation}
    Differentiating $\rho(t)=e^{A(t)}/\cZ(t)$ gives
\begin{equation}
   \ddx{t} \rho(t) 
   = \ddx{t}\left(e^{A(t)} \cdot \frac{1}{\cZ(t)}\right)
   = F(t) - \rho(t) \trace{F(t)}
,\end{equation}
where
\begin{equation}
    F(t)
    = \frac{1}{\cZ(t)} \ddx{t} e^{A(t)}
    = \integral{0}{1}{s}{ \rho(t)^s E \rho(t)^{1-s} }
.\end{equation}
By Lemma~\ref{lem:int norm inequality}, $\tnorm{F(t)} \leq \inorm{E}$.
In addition,
\begin{equation}
    \trace{F(t)} 
    = \integral{0}{1}{s}{\trace{\rho(t)^s E \rho(t)^{1-s}}}
    = \integral{0}{1}{s}{\trace{\rho(t) E }}
    = \trace{\rho(t)E} \leq \inorm{E}
.\end{equation}
Therefore $\tnorm{\ddx{t}\rho(t)} \leq 2 \inorm{E}$ for all $t\in[0,1]$, and integrating over $t$ yields
\begin{equation}
    \tnorm{\rho(1)-\rho(0)} \leq 2\inorm{E}
,\end{equation}
as claimed.
\end{proof}
\vspace{1em}

The main error bound is stated and proved below. The argument is most informative when indices repeat frequently (i.e., $N \gg \max(L,n)$), where the palindromic product can be compared sharply with an exponential via higher-order BCH control, yielding a tighter bound than a second-order Trotter decomposition~\cite{childs2018toward}.

\renewcommand\thetheorem{1}
\begin{theorem}[Complete version]~\label{thm:error appendix}
    Let $N = \BigOmega{n\hsum^2\beta^2}$ and $\delta > 0$. Then $\rho_N$ in Equation~\eqref{eqn:alg step rho} and $\gibbs{\beta}{H}$ in Equation~\eqref{eqn:expect output} satisfy
\begin{equation}
    \prob{\tnorm{\rho_N - \gibbs{\beta}{H}} > \BigO{K_0^{3/2} N^{3/2} \tau^3}} \leq \delta
\end{equation}
    for $K_0 = (n+2)\log 2 + \log (N/\delta)$.
\end{theorem}
\renewcommand{\thetheorem}{S\arabic{proposition}}
\begin{proof}
    Denote $S_k = \sum_{i=1}^k B_i$. By Lemma~\ref{lem:BCH multi expansion},
\begin{equation}
    E\coloneqq 2S_N - \log\left(\big(\prod_{k=N}^2 e^{B_k}\big) e^{2B_1} \big(\prod_{k=2}^N e^{B_k}\big)\right)
    = \frac{1}{3}\left(E_1 - 2E_2 \right) - R_N
\end{equation}
    for $E_1 = \sum_{k=2}^N\comm{B_k}{\comm{B_k}{S_{k-1}}}$, $E_2 = \sum_{k=2}^N\comm{S_{k-1}}{\comm{B_k}{S_{k-1}}}$ and some remainder term $R_N$. Lemma~\ref{lem:error bound normalize} reduces the trace-norm comparison to a spectral-norm bound on $E$:
\begin{align}
    \tnorm{\rho_N - \gibbs{\beta}{H}} 
    = \tnorm{\frac{e^{\log\big(\prod_{k=N}^2 e^{B_k}\big) e^{2B_1} \big(\prod_{k=2}^N e^{B_k}\big)}}{\trace{\textrm{same above}}} - \frac{e^{E + \log\big(\prod_{k=N}^2 e^{B_k}\big) e^{2B_1} \big(\prod_{k=2}^N e^{B_k}\big)}}{\trace{\textrm{same above}}}}
    \leq 2\inorm{E}
.\end{align}
    It therefore suffices to upper bound $\inorm{E_1}$, $\inorm{E_2}$, and $\inorm{R_N}$ with high probability, and then combine these bounds.

    We begin by controlling $\inorm{S_k}$. Let $\cF_k = (B_i)_{i=1}^k$ be the natural filtration, define the drift $D_k = \expectcond{B_k}{\cF_{k-1}}$, and let $M_k = B_k - D_k$. Then $S_k = \sum_{i=1}^k D_i + \sum_{i=1}^k M_i$. For the drift term, we have
\begin{align}
    \inorm{D_k}
    &= \inorm{\expectcond{B_k}{\cF_{k-1}}}
    = \norm{\left[ \frac{\trace{e^{\tau \sigma_{j_k}} \rho_{k-1}}}{2\cosh\tau} - \frac{\trace{e^{-\tau \sigma_{j_k}}\rho_{k-1}}}{2\cosh\tau} \right]\cdot \frac{\tau \sigma_{j_k}}{2}  }_\infty \\
    &\leq \abs{\trace{e^{\tau \sigma_{j_k}} \rho_{k-1}} - \trace{e^{-\tau \sigma_{j_k}}\rho_{k-1}} } \cdot \frac{\tau}{4\cosh\tau} \\
    &= \abs{(e^\tau - e^{-\tau}) \trace{\sigma_{j_k} \rho_{k-1}}} \cdot \frac{\tau}{2(e^\tau + e^{-\tau})} \\
    &\leq \frac{1}{2}\tau \tanh\tau = \tau^2 + \BigO{\tau^4} < 0.1\tau
,\end{align}
    for $\tau \ll 1$. The sequence $(M_k)_k$ is a martingale difference sequence of Hermitian matrices with respect to $(\cF_k)_k$, since $\expectcond{M_k}{\cF_{k-1}} = 0$.
    In particular, $\inorm{M_k} \leq \inorm{B_k} + \inorm{D_k} \leq \tau + 0.1\tau = 1.1\tau$.
    Applying Lemma~\ref{lem:freedman}, for any $k$ and $t_k > 0$,
\begin{equation}
    \prob{\inorm{\sum_{i=1}^k M_i} > t_k \textup{ and } \inorm{\sum_{i=1}^k \expectcond{M_i^2}{\cF_{i-1}}} \leq \frac{k\tau^2}{4}} \leq 2^{n + 1} \exp\!\big(-\frac{t_k^2/2}{k\tau^2/4 + 1.1\tau t_k/3}\big)
.\end{equation}
    We also have
\begin{align}
    \sum_{i=1}^k \expectcond{M_i^2}{\cF_{i-1}} 
    &= \sum_{i=1}^k \expectcond{B_i^2}{\cF_{i-1}} - D_i^2 = \sum_{i=1}^k \frac{\tau^2}{4} I - D_i^2 \\
    \implies \inorm{\sum_{i=1}^k \expectcond{M_i^2}{\cF_{i-1}}} 
    &\leq \sum_{i=1}^k \inorm{\frac{\tau^2}{4} I - D_i^2}
    \leq \sum_{i=1}^k \inorm{\frac{\tau^2}{4} I}
    =  \frac{k}{4} \tau^2
,\end{align}
    so choosing $K=(n+1)\log 2+\log \delta^{-1}$ and
    $t_k = \tau \sqrt{kK/2} + 2.2\tau K/3$ implies $(t_k/\tau)^2 / (k/2 + \frac{2.2}{3} t_k /\tau) \geq K$, and therefore
\begin{align}~\label{eqn:freedman M sum}
    \prob{\inorm{\sum_{i=1}^k M_i} > t_k } 
    &\leq 2^{n + 1} \exp\!\big( -\frac{t_k^2}{k\tau^2/2 + 2.2\tau t_k/3} \big)
    = 2^{n + 1} \exp\!\big( -\frac{(t_k/\tau)^2}{k/2 + 2.2/3\cdot (t_k/\tau)} \big) \\
    &\leq 2^{n + 1} e^{-K} = \delta
.\end{align}
    Combining the drift and martingale parts yields
\begin{equation}
    \inorm{S_k} \leq \inorm{\sum_{i=1}^k M_i} + \sum_{i=1}^k \inorm{D_i}
    \implies \prob{\inorm{S_k} > t_k + k\tau^2 } \leq \delta
.\end{equation}

    We next bound $\inorm{E_1}$. Since $B_k \sim \set{\tau\sigma_{j_k}/2, -\tau\sigma_{j_k}/2}$,
\begin{align}
    \inorm{E_1} 
    &= \norm{\sum_{k=2}^N \comm{B_k}{\comm{B_k}{S_{k-1}}}}_\infty \\
    &\leq \sum_{k=2}^N \norm{\comm{B_k}{\comm{B_k}{S_{k-1}}}}_\infty
    = \sum_{k=2}^N \inorm{ B_k^2 S_{k-1} + S_{k-1} B_k^2 - 2 B_k S_{k-1} B_k } \\
    &= \frac{\tau^2}{2} \sum_{k=2}^N \inorm{ S_{k-1} - \sigma_{j_k} S_{k-1} \sigma_{j_k} }
    \leq \tau^2 \sum_{k=2}^N \inorm{S_{k-1}} \\
    \implies &\prob{\inorm{E_1} > \tau^2 \sum_{k=2}^N (t_k + k \tau^2)} \leq \delta
.\end{align}

    The term $E_2$ is handled in a similar way. Assume $\inorm{S_k} \leq t_k + k\tau^2$ (we account for this event in the final probability bound). Define $Q_k = \comm{S_k}{\comm{M_{k + 1}}{S_k}}$, so that
\begin{equation}
    \comm{S_{k - 1}}{\comm{B_k}{S_{k-1}}} = Q_{k - 1} + \comm{S_{k-1}}{\comm{D_k}{S_{k-1}}}
\end{equation}
    with $\inorm{\comm{S_{k-1}}{\comm{D_k}{S_{k-1}}}} \leq 4 \inorm{S_k}^2 \inorm{D_{k + 1}} = 4(t_k + k\tau^2)^2 \tau^2$.
    One can verify that $\expectcond{Q_k}{\cF_k} = 0$ and
\begin{align}
    &\inorm{Q_k} \leq 4 \inorm{S_k}^2 \inorm{M_{k + 1}} = 4.4 (t_k + k\tau^2)^2 \tau,  \\
    &\inorm{\expectcond{Q_k^2}{\cF_{k-1}}} \leq 4 \inorm{S_k}^4 \inorm{M_{k + 1}}^2 \leq 20 (t_k + k\tau^2)^4\tau^2,\\
    & \inorm{\sum_{k=1}^N\expectcond{Q_k^2}{\cF_{k-1}}}\leq 20N(t_N + N\tau^2)^4\tau^2
.\end{align}
    Lemma~\ref{lem:freedman} therefore implies that, for any $k$ and $s_k > 0$,
\begin{equation}
    \prob{\inorm{\sum_{i=1}^k Q_i} > s_k} 
    \leq 2^{n + 1} \exp\!\big(-\frac{s_k^2/2}{20 N(t_N + N\tau^2)^4\tau^2 + 4.4 (t_k + k\tau^2)^2 \tau s_k/3}\big)
.\end{equation}
    Taking
    $s_k = \sqrt{40KN}(t_N + N\tau^2)^2\tau + 8.8(t_k + k\tau^2)^2K \tau/3$
    simplifies the bound to
\begin{equation}
    \prob{\inorm{\sum_{i=1}^k Q_i} > s_k} \leq 2^{n+1}e^{-K} \leq \delta
.\end{equation}
    Consequently,
\begin{align}
    \inorm{E_2} &= \inorm{\sum_{k=2}^N \comm{S_{k - 1}}{\comm{B_k}{S_{k-1}}}} 
    \leq \inorm{\sum_{k=2}^N Q_{k - 1}} + \sum_{k=2}^N\inorm{\comm{S_{k-1}}{\comm{D_k}{S_{k-1}}}} \\
    \implies & \prob{\inorm{E_2} > s_N + 4\sum_{k=2}^N (t_k + k\tau^2)^2 \tau^2} \leq \delta
.\end{align}

    We now combine the bounds. Since $N \gg n$, for all $k\leq N$ we have
    $t_k + k \tau^2\leq t_N + N\tau^2=\BigO{\tau \sqrt{NK}} + N\tau^2=\BigO{\tau \sqrt{NK}}$.
    Therefore, $\inorm{E_1}\leq \tau^2 \sum_{k=2}^N (t_k + k \tau^2) = \BigO{N\tau^3\sqrt{NK}}$ and
    $\inorm{E_2} \leq s_N + 4\sum_{k=2}^N (t_k + k\tau^2)^2 \tau^2 = \BigO{\tau^3 K^{3/2} N^{3/2}}$,
    which implies $\inorm{E_1-2E_2} = \BigO{\tau^3 K^{3/2} N^{3/2}}$.

    It remains to bound the remainder term. By Lemma~\ref{lem:BCH multi expansion},
\begin{equation}
     \inorm{\cR_{N}} =\BigO{\sum_{r=5, r \textrm{ is odd}}^\infty \frac{1}{r!}\sum_{k=2}^{N} \inorm{\operatorname{ad}^{r-1}_{2S_{k-1}}(B_k)}}.
\end{equation}
For an odd order $r\geq 5$, applying Lemma~\ref{lem:freedman} yields
\begin{equation}
    \frac{1}{r!}\sum_{k=2}^{N} \inorm{\operatorname{ad}^{r-1}_{2S_{k-1}}(B_k)} = \frac{2^{r}}{r!}\BigO{(\sqrt{NK} \tau)^r} = o(\tau^3 K^{3/2} N^{3/2}).
\end{equation}
Alternatively, the triangle inequality gives
\begin{equation}
    \frac{1}{r!}\sum_{k=2}^{N} \inorm{\operatorname{ad}^{r-1}_{2S_{k-1}}(B_k)} = \frac{2^{2r}}{r!}\BigO{N\tau(\sqrt{NK} \tau)^{r-1}}. 
\end{equation}
For $r^\star = \left\lceil \frac{\log(K\tau^2 \sqrt{NK})}{\log(\tau\sqrt{NK})}\right\rceil$, this implies $\inorm{R_N}=\BigO{\tau^3 K^{3/2} N^{3/2}}$.
We therefore apply Lemma~\ref{lem:freedman} for $r < r^\star$ and use the triangle inequality for $r \geq r^\star$; the remaining tail ($r>r^\star$) decays exponentially and can be neglected at the stated precision.

    Substituting these bounds into $\tnorm{\rho_N - \gibbs{\beta}{H}} \leq 2\inorm{E}$ yields
\begin{equation}
    \prob{\tnorm{\rho_N - \gibbs{\beta}{H}} \leq \BigO{K^{3/2} N^{3/2} \tau^3}} \geq 1-\delta.
\end{equation}

    It remains to account for the overall failure probability. The proof applies Lemma~\ref{lem:freedman} $(r^\star -1)/2 + N$ times. Setting $\delta'=\delta/((r^\star -1)/2 + N)$ and taking a union bound ensures that the total failure probability is at most $\delta$.
    Under this rescaling,
    $K = (n+1)\log 2 + \log \bigl(((r^\star -1)/2 + N )/\delta\bigr)$.
    Given $N\gg n$ and $\tau\sqrt{NK} \ll 1$, we may assume $N\gg K$ and obtain
\begin{align}
    r^\star \leq 2 + \frac{\log (K\tau)}{\log(\tau\sqrt{NK})} \ll N,
\end{align}
which implies
$K \leq (n+2)\log 2 + \log (N/\delta) \ll N$.
Taking $K_0 = (n+2)\log 2 + \log (N/\delta)$ therefore yields the stated bound.
\end{proof}

\subsection{Sample distribution}

Recall that Algorithm~\ref{alg:sample thermal} generates a path $\cP = \big((j_1,m_1), \ldots, (j_N,m_N)\big) \in [L] \times \set{1, -1}$.
In Section~\ref{sec:error}, we assumed $\cP$ is fixed so that the evolution sequence can be analyzed without randomness. This section uses a random walk to incorporate the randomness of $j_k$ and $m_k$.

Consider how the sampled Hamiltonian evolves as the path develops. Stopping the algorithm at step $k \in [N]$ yields the intermediate Hamiltonian $\walkH{k} = \sum_{t=1}^k m_t \sigma_{j_t}$. The sampling process can thus be described as a sequence $(\walkH{k})_k$, with $\walkH{N}$ being the final output. Equivalently, this sequence is a random walk on a lattice:
\begin{equation}~\label{eqn:walk def}
    \walk{k} = \sum_{t=1}^k m_t e_{j_t} \in \set{-N, \ldots, N}^{\times L},
    \quad
    \walk{0} = 0
,\end{equation}
where $e_{j_t}$ is the standard vector in the lattice corresponding to $\sigma_{j_t}$,
and every forward step satisfies
\begin{equation}
    \prob{\walk{k} - \walk{k-1} = \pm e_j} = \prob{m_k = \pm 1} \cdot \bmh_j /\hsum
.\end{equation}
The sign of $m_k$ is determined by the quantum measurement in Theorem~\ref{thm:pauli select}, which is difficult to analyze directly.
We therefore introduce a reference probability measure $Q$ with $\probQ{m_k = \pm 1} = 1/2$, giving 
\begin{equation}~\label{eqn:reference prob}
    \probQ{\walk{k} - \walk{k-1} = \pm e_j} = \bmh_j /2\hsum
\end{equation}
corresponding to $N \to \infty$. The relation between the reference probability measure and the true probability measure is stated by the likelihood ratio $L_N(\cP) = \prob{\cP} / \probQ{\cP}$ in the following statement.

\begin{lemma}~\label{lem:independ on path}
    If $N = \BigOmega{n\hsum^2\beta^2}$, then for $H = H^{(N)}$,
\begin{equation}
    \log L_N(\cP)
  =\log\!\left(\frac{1}{2^n}\trace{e^{-\beta H}}\right)
   +\BigO{\frac{\hsum^2\beta^2}{N}}.
\end{equation}
\end{lemma}
\begin{proof}
For convenience, denote $\rho_\bfx$ as the thermal state of Hamiltonian with respect to $\bfx$ and $M_j(\bfx) = \trace{\sigma_j \rho_{\bfx}}$
By Theorem~\ref{thm:pauli select}, the outcome distribution at step $t$ given the Pauli operator $\sigma_{j_t}$ can be written as 
\begin{equation}
    \prob{m_t=\pm1 \mid j_t} = \frac{1}{2}\Bigl(1\mp \tanh\tau \trace{\sigma_{j_t} M_j(\walk{t-1}) }\Bigr)
\end{equation}
Since $\probQ{m_t = \pm1 \mid j_t} = 1/2$ and the probabilities for sampling $j_t$ for both measures are the same, the likelihood ratio takes the product form
\begin{equation}
     L_N(\cP) 
     = 2^N \prod_{t=1}^N  \prob{m_t=\pm1 \mid j_t}
     = \prod_{t=1}^N  \Bigl(1 - m_t \tanh(\tau) M_{j_t}(\walk{t-1})\Bigr)
.\end{equation}
Using $\log(1+y)=y+\BigO{y^2}$ and $\tanh\tau=\tau+\BigO{\tau^3}$ gives
\begin{equation}\label{eqn:logLN-expand}
  \log L_N(j,m)
  = -\tau\sum_{t=1}^N m_t M_{j_t}(\walk{t-1})+\BigO{N\tau^2}
.\end{equation}
Denote the parition function of $H^{(t)}$ at the $t$-th step as
\begin{equation}
    \cZ(\walk{t}) 
    = \trace{\exp\!\left(-\beta H^{(t)} \right)}
    = \trace{\exp\!\left(-\tau \sum_{j=1}^L \walk{t}_j \sigma_j \right)} 
.\end{equation}
Observe that differentiating $\log \cZ(\bfx)$ with respect to $\bfx_j$ gives
\begin{equation}
  \partial_{\bfx_j} \log\cZ(\bfx)
  =\frac{1}{\cZ(\bfx)}\trace{\left(-\tau\sigma_j\right) \exp\!\left(-\tau\sum_{j=1}^L \bfx_j\sigma_j\right)}
  =-\tau M_j(\bfx)
.\end{equation}
Therefore, since $\walk{t}-\walk{t-1}=m_t e_{j_t}$,
\begin{equation}
  -\tau m_t M_{j_t}(\walk{t-1})
  =m_t \partial_{\walk{t-1}} \log\cZ(\walk{t-1})
  =\nabla_\bfx \log\cZ(\walk{t-1})\cdot\bigl(\walk{t}-\walk{t-1}\bigr)
.\end{equation}

Because $\cZ(\bfx)$ depends on $\bfx$ through $\tau \bfx$ and $\tau\ll 1$, a Taylor
expansion of $\log\cZ$ along one lattice step gives
\begin{equation}
  \log\cZ(\walk{t})-\log\cZ(\walk{t-1})
  =\nabla_\bfx \log\cZ(\walk{t-1})\cdot\bigl(\walk{t}-\walk{t-1}\bigr)
   +\BigO{\tau^2}
.\end{equation}
Summing over $t=1,\dots,N$ and substituting into Equation~\eqref{eqn:logLN-expand} yields
\begin{equation}
  \log L_N(j,m)
  =\log\cZ(\walk{N})-\log\cZ(0)+\BigO{N\tau^2}.
\end{equation}
Finally, $\cZ(0)=\trace{I}=2^n$ and
$\cZ(\walk{N})=\trace{\exp(-\tau\sum_j \bfx_j\sigma_j)}=\trace{e^{-\beta H^{(N)}}}$, so
the statement follows from $N\tau^2=\hsum^2\beta^2/N$.
\end{proof}
\vspace{1em}

We now state the theoretical result for the sample distribution.

\renewcommand\thetheorem{2}
\begin{theorem}[Complete version]~\label{thm:distribution appendix}
    Suppose $n \gg 1$ and $N = \BigOmega{n\hsum^2\beta^2}$. Let $H = \sum_{j=1}^L c_j \sigma_j$ and $p_j = c_j /\hsum$.
  When $\sum_j c_j^2 \leq N \hsum^2$, the datum $(\gibbs{\beta}{H}, H)$ can be sampled by Algorithm~\ref{alg:sample thermal} with probability proportional to $\frac{1}{2^{n}} \trace{e^{-\beta H}} \cdot D(\bfx)$,
  where $D$ is given as
\begin{equation}
    D(\bfx) = \frac{a_N(\bfx)}{(2\pi N)^{L/2}\sqrt{\prod_{j=1}^L p_j}}
    \exp\left(-\frac{1}{2N} \sum_j p_j \bfx_j^2 \right)
    + o\bigl(N^{-L/2}\bigr)
.\end{equation}
\end{theorem}
\renewcommand{\thetheorem}{S\arabic{proposition}}
\begin{proof}
Fix $H=\sum_{j=1}^L c_j\sigma_j$ on the output grid, and let
$\bfx \in\ZZ^L$ be the corresponding endpoint vector $\bfx_j = c_j N / \hsum$. Denote $\cP_\bfx$ the set of paths whose endpoint equals $\bfx$.
Then the event that Algorithm~\ref{alg:sample thermal} returns $H$ is exactly the
endpoint event $\walk{N}= \bfx$.

By definition of the likelihood ratio and Lemma~\ref{lem:independ on path},
\begin{equation}
  \prob{\cP}
  = \probQ{\cP}\cdot \exp\left(
  \log\left(\frac{1}{2^n}\trace{e^{-\beta H}}\right)
  +\BigO{\frac{\hsum^2\beta^2}{N}}\right).
\end{equation}
When $\norm{\bfx}=\BigO{\sqrt{N}}$. Therefore, absorbing the $\bfx$-independent normalization into the proportionality constant, summing over all $\cP$ with same endpoint, we obtain
\begin{align}
    \prob{\walk{N}=\bfx}&= \exp\left(
  \log\left(\frac{1}{2^n}\trace{e^{-\beta H}}\right)
  +\BigO{\frac{\hsum^2\beta^2}{N}}\right)\cdot \sum_{\cP\in\cP_\bfx}\probQ{\cP}\notag\\ 
    &\propto \frac{1}{2^{n}}\tr\left[e^{-\beta H}\right]\cdot
  \probQ{\walk{N}=\bfx}
.\end{align}
Denote $D(\bfx) = \probQ{\walk{N}=\bfx}$. Substituting Lemma~\ref{lem:reference distribution} yields
the Gaussian approximation claimed in the theorem.
\end{proof}

\section{Application details}\label{sec:apps}

\subsection{Level statistics of quantum density operators}

This section outlines the framework for analyzing the spectral properties of the generated quantum thermal states, which we use as a benchmark for whether the sampler captures nontrivial chaotic correlations. Level statistics, commonly applied to Hamiltonian energy spectra to diagnose quantum chaos and ergodicity~\cite{wigner1967random,oganesyan2007localization}, extend naturally to density operators $\rho$ via the modular Hamiltonian. 

For a strictly positive density operator $\rho$, its spectral properties are closely related to those of the fictitious Hamiltonian governing the thermal distribution. We define the modular Hamiltonian $H$ by $\rho = \frac{e^{-H}}{\tr(e^{-H})}$. The eigenvalues of $\rho$, denoted by $\{p_i\}$, are directly related to the spectrum of $H$, denoted by $\{\xi_i\}$, via the transformation $\xi_i = -\ln(p_i) - \ln Z$. The level statistics of $\rho$ are thus analyzed through the unfolding of the "entanglement energies" or modular levels $\{\xi_i\}$.

To avoid the ambiguities associated with spectral unfolding in finite-size systems, we utilize the adjacent gap ratio statistics introduced by Oganesyan and Huse~\cite{oganesyan2007localization}. Let the sorted spectrum of the modular Hamiltonian be ordered such that $\xi_1 \le \xi_2 \le \dots \le \xi_D$, where $D$ is the dimension of the Hilbert space. We define the nearest-neighbor level spacing as:
\begin{equation}
    \delta_n = \xi_{n+1} - \xi_n \geq 0
.\end{equation}
The dimensionless level spacing ratio $r_n$ is defined as:
\begin{equation}
    r_n = \frac{\min(\delta_n, \delta_{n+1})}{\max(\delta_n, \delta_{n+1})} = \min\left( \tilde{r}_n, \frac{1}{\tilde{r}_n} \right) \in [0,1]
\end{equation}
where $\tilde{r}_n = \delta_{n+1} / \delta_n$. The distribution of $r_n$, denoted as $P(r)$, and its mean value $\langle r \rangle$, serve as robust indicators of the nature of the correlations in the quantum state. The following summarizes the known results of the average spacing ratios with respect to the specific RMT ensembles~\cite{atas2013distribution}:
\begin{itemize}
    \item \emph{Integrable (Poissonian) limit}: Systems with uncorrelated levels, typical of integrable models or localized phases, exhibit level spacings following a Poisson distribution with $P(r) = \frac{2}{(1+r)^2}$ and $\langle r \rangle_{\text{Poisson}} \approx 0.386$.
    \item \emph{Chaotic (Wigner--Dyson) limit}: Ergodic quantum systems exhibit level repulsion characteristic of Random Matrix Theory (RMT), with statistics depending on the symmetry class: $\langle r \rangle_{\text{GOE}} \approx 0.536$ for systems with time-reversal symmetry, and $\langle r \rangle_{\text{GUE}} \approx 0.603$ for systems without.
\end{itemize}
In our algorithm, $\langle r \rangle$ values approaching the GUE/GOE limits in the sampled thermal states indicate that the generated state captures the complexity and correlations of a thermalizing many-body system. Conversely, $\langle r \rangle \approx 0.39$ would suggest a failure to capture thermal correlations or an underlying integrability in the generator set.

\subsection{State classification by Hamiltonian properties}~\label{sec:state classification}

Let us specify the background first. 
Let $\Sigma$ and $\bmh$ be as defined in Section~\ref{sec:sample proof} that describe a set of Hamiltonians $\fH$.
For a label function $f: \fH \to \set{0, 1}$ that divide this set to group 0 and group 1, given access to an unknown thermal state $\gibbs{\beta}{H}$ (with known temperature $\beta$), the goal of state classification is to predict whether $f(H)$ is 0 or 1.
Denote the labelled set $D_y = \setcond{\rho}{(\rho, H) \in \dataset,\, f(H) = y}$ for $y = 0, 1$, and without loss of generality, suppose $\abs{D_0} = \abs{D_1}$.
Then, equivalently, we are looking for an $n$-to-1 channel $\cE$ that maximizes the average success probability of prediction
\begin{equation}~\label{eqn:target prob}
    P(\cE) = \frac{1}{2} \sum_{y = 0}^1 \prob{\textrm{predict }y \mid \rho \in D_y}
    = \frac{1}{2} \sum_{y = 0}^1 \integral{D_y}{}{ \rho}{ \trace{\cE(\rho) \cdot \ketbra{0}{0}}}
    = \frac{1}{2} \sum_{y = 0}^1 \trace{\cE(\rho_y) \cdot \ketbra{y}{y}}
\end{equation}
for $\rho_y = \integral{D_y}{}{ \rho}{ \rho }$ is the average mixture of the state set $D_y$, and $\rho_y = \frac{1}{\abs{D_y}} \sum_{\rho \in D_y} \rho$ when $D_y$ is finite. Theoretically, this average success probability of prediction is upper bounded by the Helstrom limit~\cite{helstrom1969quantum},
\begin{equation}~\label{eqn:helstrom}
    P(\cE) \leq \frac{1}{2} (1 + \tnorm{\rho_0  - \rho_1})
.\end{equation}
In later experiment, the Helstrom limit is used to evaluate the convergence performance.

The rest of this section outlines an example training algorithm on how to use our sampler as data source to find a protocol to tackle this task. In the standard expectation-value-based formulation, each loss evaluation typically requires many circuit repetitions on the same input state; here we instead use single-shot outcomes on freshly sampled batches from the stochastic data source.
We model the $n$-to-1 channel $\cE$ by an $(n+1)$-qubit parameterized quantum circuit~\cite{benedetti2019parameterized} $U(\theta)$ with one ancilla qubit. The circuit takes an $n$-qubit state $\rho$ as input and traces out the main system; we denote the resulting parameterized channel as $\cE_\theta$.
Then for a finite set of input state $\rho_1 , \ldots,\rho_d$ that approximates thermal states $\gibbs{\beta}{H_1}, \ldots, \gibbs{\beta}{H_d}$ generated by our sampler, respectively,
and their 0/1 label vector $\hat{y} = (f(H_1), \ldots, f(H_d)) \in \set{0, 1}^d$, the average prediction loss is defined as the stochastic function
\begin{equation}~\label{eqn:loss}
    \loss(\theta) = \big(1 - \frac{1}{d} y(\theta) \cdot \hat{y}\big) /2,
    \textrm{\quad where\,\,} \prob{y_i(\theta) = \hat{y}_i} = \trace{\cE_\theta(\gibbs{\beta}{H_i}) \cdot \ketbra{f(H_i)}{f(H_i)}}
\end{equation}
such that the expectation is the average failure probability of prediction,
\begin{equation}
    \expect{\loss(\theta)} 
    = 1 - \frac{1}{d}\sum_{i=1}^d\trace{\cE_\theta(\gibbs{\beta}{H_i}) \cdot \ketbra{f(H_i)}{f(H_i)}}
.\end{equation}
Then our objective is to find the optimal parameters $\theta^* = \arg \min_\theta \expect{\loss(\theta)}$.
Since $\loss(\theta)$ is stochastic, we estimate the gradient using a stochastic variant of the parameter-shift rule~\cite{mitarai2018quantum,schuld2019evaluating}. Unlike the standard formulation based on expectation-value estimation, each loss evaluation here is obtained from single-shot measurements on a freshly sampled batch.
For each parameter $\theta_j$ associated with a Pauli-rotation gate $U(\theta_j) = \exp(-i\theta_j P/2)$, the gradient component can be estimated by
\begin{equation}
    \nabla_j {\loss}(\theta) = \frac{\loss(\theta + \frac{\pi}{2} e_j) - \loss(\theta - \frac{\pi}{2} e_j)}{2},
\end{equation}
where $e_j$ is the unit vector in the $j$th coordinate. This estimator is unbiased in the presence of a stochastic data source, i.e., $\expect{\nabla {\loss}(\theta)} = \nabla_\theta \expect{\loss(\theta)}$.
The complete training procedure is summarized in Algorithm~\ref{alg:psr}. Here $\epsilon$ controls the sampling precision of our sampler, and $w$ is a weight function that controls the importance of each Hamiltonian.

\begin{algorithm}[H]\label{alg:psr}
\caption{Training a parameterized quantum circuit for state classification}
\SetKwInOut{Input}{Input}
\SetKwInOut{Output}{Output}
\Input{A label function $f$, inverse temperature $\beta$, a set $\Sigma$ of $n$-qubit Pauli operators with boundary vector ${\bmh}$, precision $\epsilon$, training steps $T$, batch size $d$, learning rate $\gamma$, weight function $w(H) \in \RR_+$}
\Output{A channel $\cE$ that approximately maximizes $P(\cE)$ in Equation~\eqref{eqn:target prob}}

$N \gets \hsum^2 \beta^2 \epsilon^{-2/3}$\;

$\mathrm{SAMPLE} \gets$ Algorithm~\ref{alg:sample thermal} with inputs $\beta$, $\Sigma$, $\bmh$, $N$\;

Choose a circuit ansatz $U(\theta)$ with random parameters $\theta$, and define $\cE_\theta \gets \trace[\textrm{sys}]{U(\theta) \left( \ketbra{0}{0}_{\textrm{anc}} \ox (\cdot)_{\textrm{sys}} \right) U(\theta)^\dagger}$\;

\For{$t = 1, \ldots, T$}{
    $g \gets 0^{\abs{\theta}}$\;
    \For{$j = 1, \ldots, \abs{\theta}$}{
        \For{$s \in \set{\mathord{+}1, \mathord{-}1}$}{
            \For{$i = 1, \ldots, d$}{
                $(\rho_i, H_i) \gets \mathrm{SAMPLE}()$\;
                $y_i \gets$ measure $\cE_{\theta + s \cdot \frac{\pi}{2} e_j}(\rho_i)$, $\hat{y}_i \gets f(H_i)$\;
            }
            $\loss_s \gets \big(1 - \sum_i w(H_i) \, y_i \hat{y}_i \big/ \sum_i w(H_i) \big) /2$\;
        }
        $g_j \gets (\loss_{+1} - \loss_{-1})/2$\;
    }
    $\theta \gets \theta - \gamma \cdot g$ \quad (or use Adam optimizer)\;
}
\Return{$\cE = \cE_\theta$}
\end{algorithm}

The convergence of Algorithm~\ref{alg:psr} depends on the variance of the parameter-shift gradient estimate. Since each iteration requires $2\abs{\theta}$ stochastic loss evaluations, both the evaluation budget and the batch size $d$ affect the gradient noise in the presence of data and measurement sampling. If $d$ is not large enough, the variance may be too high for the classifer to converge. The weight function $w$ adjusts the loss in Equation~\eqref{eqn:loss} to emphasize Hamiltonians whose importance is not too small, thereby improving the effective signal-to-noise ratio.

\section{Experimental setting}\label{sec:exp_setting}

\subsection*{Figure 2-3a} 
All experiments for Figure 2 and Figure 3a are performed on a $3 \times 3$ two-dimensional grid of $n = 9$ qubits, with sites labeled $1, \ldots, 9$ in row-major order. The nearest-neighbor pairs $\langle i, j \rangle$ consist of $12$ edges connecting horizontally and vertically adjacent sites. For the Heisenberg model, the sampled Hamiltonians take the form
\begin{equation}
    H_{\mathrm{Heis}} = \sum_{\langle i,j \rangle} \left( c_{ij}^{(X)} X_i X_j + c_{ij}^{(Y)} Y_i Y_j + c_{ij}^{(Z)} Z_i Z_j \right),
\end{equation}
where the coefficients $c_{ij}^{(\alpha)}$ are generated by the sampling algorithm and satisfy $|c_{ij}^{(\alpha)}| \leq h$ for a uniform bound $h > 0$. For the transverse-field Ising model, the sampled Hamiltonians take the form
\begin{equation}
    H_{\mathrm{Ising}} = \sum_{\langle i,j \rangle} c_{ij} Z_i Z_j + \sum_{i=1}^{n} g_i X_i,
\end{equation}
where the coefficients $c_{ij}$ and $g_i$ are similarly bounded. In both cases, $\hsum = \sum_j h_j$ denotes the sum of all coefficient bounds. Throughout, we set the step-number constant $C = \hsum^{2}/\epsilon_0^{2/3}$ with target precision $\epsilon_0 = 10^{-6}$.

For Figure~2, all three panels use the Heisenberg model with the maximally mixed state $I/d$ as the initial state. In panel~(a), we scan the inverse temperature from $\beta = 1$ to $\beta = 10$ and test three scaling exponents $k \in \set{1.5, 2.0, 2.5}$, which define the step number as $N = C \beta^{k}$. For each parameter pair $(\beta, N)$, we perform $100$ independent sampling runs and record the maximal trace-distance error $\epsilon$ among the samples. Panels~(b) and~(c) both fix $\beta = 2$. In panel~(b), we set $N = C \beta^{2}$ and generate $10\,000$ independent samples to obtain the empirical coefficient distribution; the theoretical curve is computed from the normal distribution predicted by Theorem~2, reweighted by the thermal partition function and estimated via Monte Carlo with $10\,000$ samples. In panel~(c), we scan the scaling exponent $k$ from $1$ to $3$ and perform $100$ independent runs for each $N = C \beta^{k}$. The inverse precision is defined as $1/\epsilon$, where $\epsilon$ is the maximal trace-distance error among the samples. The effective sample range is quantified by the operator norm $\inorm{H}$ of each sampled Hamiltonian; we report the sample mean and standard error over the thermal ensemble.

For Figure~3(a), we use the transverse-field Ising model to generate the output states. The initial state is a thermal state of a Heisenberg Hamiltonian, whose modular spectrum exhibits near-Poisson statistics in this finite-size setting. We fix $\beta = 2$ and $N = C \beta^{2}$, and perform $20$ independent sampling runs. For each output state $\rho$, we define the modular Hamiltonian $K = -\log \rho$ and compute the adjacent gap-ratio statistic $r$ from the spectrum of $K$. To handle near-degeneracies, eigenvalues within a relative tolerance of $10^{-6}$ are merged into a single representative level before evaluating the spacing ratios. The resulting $r$ values are aggregated over all samples to form the final histogram.

\subsection*{Figure 3b}

\begin{figure}[H]
    \centering
    \includegraphics[width=0.5\linewidth]{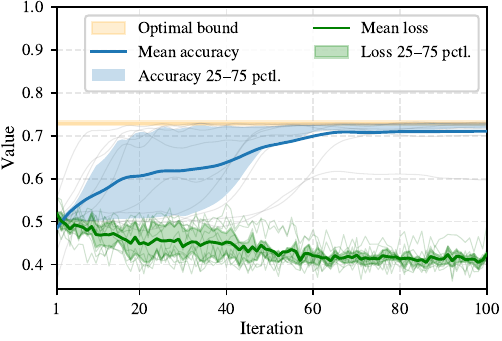}
    \caption{Training details for the state classification task, showing both accuracy and loss as functions of the training iteration.}
    \label{fig:classification}
\end{figure}

The complete figure of Figure 3b, including loss values, is shown in Figure~\ref{fig:classification}.
In this experiment, we follow the procedure described in Section~\ref{sec:state classification} to classify two-qubit Hamiltonians by the sign of their $X \ox X$ coefficient, using their thermal states at inverse temperature $\beta = 2$.
The Hamiltonians take the form $H = h_1 X \ox X + h_2 I \ox Z + h_3 Z\ox I$ with $\abs{h_1}, \abs{h_2}, \abs{h_3} \leq 1$. 
We further restrict to Hamiltonians whose $X \ox X$ coefficient satisfies $\abs{h_1} \geq 0.03$.
We run Algorithm~\ref{alg:psr} over $10$ independent trials with the following setting:
\begin{enumerate}[leftmargin=1em]
\setlength{\itemsep}{0pt}
\setlength{\parskip}{0pt}
    \item [-] $f: H \mapsto h_1 / \abs{h_1}$,
    \item [-] $\Sigma = \set{X \ox X, I \ox Z, Z\ox I}$ with $h = (1, 1, 1)$,
    \item [-] number of training step $T = 100$, number of training size $d = 1000$,
    \item [-] parameter-shift value $\pi/2$, learning rate $\gamma = 0.1$,
    \item [-] the weight function $w(H) = 0.05 + 0.95 \operatorname{Sigmoid}((\abs{h_1} - 0.01\lambda) / 0.0025\lambda)$.
\end{enumerate}
We choose the implementation of the variational quantum classifier to consist of a universal $2$-qubit layer (built from parameterized single-qubit universal gates $U$ and Pauli-rotation gates $R_z$, $R_y$) on the input state, together with two additional single-qubit universal gates on the ancilla before and after coupling it to the system, shown as follows.

\begin{figure}[H]
    \centering
    \includegraphics[width=\linewidth]{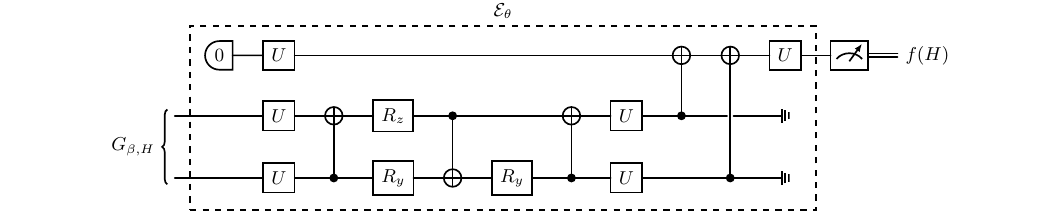}
% \begin{quantikz}[transparent]
% {} & \inputD{0}\gategroup[3,steps=11,style={inner sep=4pt,dashed,label={above:{$\cE_\theta$}}}]{} & \gate[1]{U} & {} & {} & {} & {} & {} & {} & \targ{} & \targ{} & \gate[1]{U} & \meter{} & \setwiretype{c}\rstick{$f(H)$} \\
% \lstick[2]{$\gibbs{\beta}{H}$} & {} & \gate[1]{U} & \targ{} & \gate[1]{R_{z}} & \ctrl[]{1} & {} & \targ{} & \gate[1]{U} & \ctrl[]{-1} & \push{\,\,} & \ground{}  \\
% {} & {} & \gate[1]{U} & \ctrl[]{-1} & \gate[1]{R_{y}} & \targ{} & \gate[1]{R_{y}} & \ctrl[]{-1} & \gate[1]{U} & {} & \ctrl[]{-2} & \ground{}
% \end{quantikz}
\end{figure}

For each independent trial, we also prepare a separate test set of size $3000$ to evaluate the model at each training iteration. For each test set, we compute the Helstrom limit in Equation~\eqref{eqn:helstrom}, which yields an upper bound of approximately $0.72$--$0.74$.
To assess generalization, the test set consists of exact thermal states with Hamiltonian coefficients $h_1$, $h_2$, $h_3$ sampled from a uniform distribution, in contrast to the training set whose approximate thermal states are sampled from the thermal-weighted normal distribution (Theorem~\ref{thm:distribution appendix}).
The performance metric is the average success probability of prediction defined in Equation~\eqref{eqn:target prob}.
 
\end{document}